\documentclass[showpacs,amsmath,amssymb,twocolumn,pra,superscriptaddress,notitlepage]{revtex4-2}

\usepackage[dvips]{graphicx} 
\usepackage{amsfonts}
\usepackage{amssymb}
\usepackage{amscd}
\usepackage{enumerate}
\usepackage{epsfig}
\usepackage{subfigure}
\usepackage{xcolor}
\usepackage{amsthm}
\usepackage{framed}
\usepackage{multirow}
\usepackage{braket}
\usepackage[colorlinks = true]{hyperref}

\newtheorem{theorem}{Theorem}
\newtheorem{lemma}{Lemma}

\newtheorem{definition}{Definition}

\newtheorem{proposition}{Proposition}

\newtheorem{fact}{Fact}
\newcommand{\ketbra}[2]{\vert #1 \rangle \langle #2 \vert}
\newcommand{\bk}[1]{\mbox{$\left\langle #1 \right\rangle$}}

\newcommand{\comments}[1]{}

\newcommand{\id}{\mathbb{I}}

\newcommand{\Tr}[1]{\mbox{$\text{Tr}\left( #1 \right)$}}

\begin{document}
\preprint{APS/123-QED}
\title{Quantum subspace verification for error correction codes}

\date{\today}
\author{Junjie Chen}
\affiliation{Center for Quantum Information, Institute for Interdisciplinary Information Sciences, Tsinghua University, Beijing 100084, China}
\author{Pei Zeng}
\affiliation{Pritzker School of Molecular Engineering, The University of Chicago, Chicago 60637, USA}
\author{Qi Zhao}
\affiliation{Department of Computer Science, University of Hong Kong, Pokfulam Road, Hong Kong}
\author{Xiongfeng Ma}
\affiliation{Center for Quantum Information, Institute for Interdisciplinary Information Sciences, Tsinghua University, Beijing 100084, China}
\author{You Zhou}
\email{you\_zhou@fudan.edu.cn}
\affiliation{Key Laboratory for Information Science of Electromagnetic Waves (Ministry of Education), Fudan University, Shanghai 200433, China}

\begin{abstract}
Benchmarking the performance of quantum error correction codes in physical systems is crucial for achieving fault-tolerant quantum computing. Current methodologies, such as (shadow) tomography or direct fidelity estimation, fall short in efficiency due to the neglect of possible prior knowledge about quantum states. To address the challenge, we introduce a framework of quantum subspace verification, employing the knowledge of quantum error correction code subspaces to reduce the potential measurement budgets. Specifically, we give the sample complexity to estimate the fidelity to the target subspace under some confidence level. Building on the framework, verification operators are developed, which can be implemented with experiment-friendly local measurements for stabilizer codes and quantum low-density parity-check (QLDPC) codes. Our constructions require $O(n-k)$ local measurement settings for both, and the sample complexity of $O(n-k)$ for stabilizer codes and of $O((n-k)^2)$ for generic QLDPC codes, where $n$ and $k$ are the numbers of physical and logical qubits, respectively. Notably, for certain codes like the notable Calderbank-Shor-Steane codes and QLDPC stabilizer codes, the setting number and sample complexity can be significantly reduced and are even independent of $n$. In addition, by combining the proposed subspace verification and direct fidelity estimation, we construct a protocol to verify the fidelity of general magic logical states with exponentially smaller sample complexity than previous methods. Our finding facilitates efficient and feasible verification of quantum error correction codes and also magical states, advancing the realization in practical quantum platforms.
\end{abstract}

\maketitle

\section{Introduction}

Quantum error correction (QEC) is pivotal in advancing towards large-scale quantum computation. Numerous QEC codes are constructed, such as surface codes \cite{Kitaev2003Toric, Bravyi1998Quantum, Google2024surfacecodethreshold}, good quantum low-density parity-check (QLDPC) codes \cite{Gottesman2014Faulttolerant, Panteleev2022GoodLDPC, xu2024constant}, and random stabilizer codes \cite{Gullans2021RandomQEC}, each offering unique advantages and compromises. Benchmarking the performance of various QEC codes across different physical systems is crucial for identifying codes compatible with fault-tolerant quantum computing on these systems \cite{Finsterhoelzl2023benchmarkingQEC}. The performance in general can be characterized by the fidelity between the physical state during quantum computing and the QEC code subspace. However, conventional methods such as (shadow) tomography \cite{cramer2010efficient, gross2010quantum, Lvovsky2009CVtomoRMP, Huang2020Predicting, Senrui2021Robust,zhou2024hybrid} and fidelity estimation \cite{Flammia2011Direct,Guhne2007toolbox,Seshadri2024versatile} are resource-intensive. They suffer from an exponential growth in resource demands due to failing to utilize prior knowledge of the quantum states.

An efficient framework, quantum state verification, was introduced to leverage prior knowledge and reduce the resources required \cite{Sam2018Optimal,hayashi2006study,Zhu2019Efficient, Zhu2019General, yu2019optimal,Liu2019Dicke,Liu2021Nondemolition}. This framework involves measuring the prepared noisy states over multiple rounds, and conditioning on passing all tests in each round to lower bound the fidelity to the ideal pure state within some specified confidence level. It has been demonstrated that state verification can be efficiently executed for all stabilizer states, with only local measurements \cite{Dangniam2020Stabilizer, Liu2019Dicke, Huang2024Certifying}. However, it would complicate the situation if one naively applies the state verification to the QEC problems here. 
Such a strategy requires the design of different protocols tailored to the target physical states in the QEC code subspace. Moreover, there is a lack of efficient verification protocols for general magic states, out of reach of stabilizer formalism \cite{Zhu2019Hypergraph}, especially for logical magic states.

To advance the QEC verification, here we extend the state verification to that of quantum subspace, utilizing the prior knowledge of the code subspace, making it suitable for this challenge. Modeled as a hypothesis testing problem, we derive the number of measurements sufficient to discriminate whether the fidelity to the subspace exceeds some specified threshold. Moreover and importantly, we can also estimate the range of such fidelity directly from the collected measurement results beyond the state-verification scenario \cite{Zhu2019Efficient, Zhu2019General}. Indeed, the design of the measurement strategy - essentially, the corresponding verification operator  - is crucial as both the sample complexity and the range of estimated fidelity are closely related to its spectral property.

Based on the established framework, we explore the constructions of verification operators for two prevalent classes of QEC codes: stabilizer codes and QLDPC codes, using local measurements, which facilitate practical realizations in physical systems \cite{Google2024surfacecodethreshold, xu2024constant}. In this context, we employ a generalized concept of QLDPC codes, which are defined based on projectors and do not necessarily conform to traditional stabilizer codes. For stabilizer codes, we adopt a graphical method to assemble the stabilizers on account of their commuting relations into subsets, which can even reduce the number of local measurement settings and sample complexity to constant for like Calderbank-Shor-Steane (CSS) codes \cite{Calderbank1996CSScode, Steane1996CSScode} and QLDPC stabilizer codes \cite{Gottesman2014Faulttolerant, Panteleev2022GoodLDPC}. For generic $[[n,k,d]]$ QLDPC codes, if only single-qubit measurement is permitted, the maximal eigenvalues of the verification operator would be less than $1$ which is beyond the situations of the previous state-verification problems \cite{Zhu2019Efficient, Zhu2019General}. A proposition is given to deal with this situation, which indicates that the sample complexity can be bounded by $O(n-k)^2$. In addition to the verification of QEC codes, our framework offers potential in the realm of quantum state learning \cite{Anshu2024Survey}. With the adoption of advanced subspace verification to reduce the potential space of the quantum states, we can further apply direct fidelity estimation on the logical subspace, which can reduce the sample complexity exponentially and facilitate the magic-state verification. 

\section{Framework of subspace verification}\label{sec:frame}

In this section, we present a framework for subspace verification. It is important to note that this framework applies to any quantum subspace verification and is not solely restricted to QEC code subspaces.

Given a quantum device $D$ that should produce quantum states in a specific subspace $\mathcal{V}$, our objective is to verify whether it works correctly. Suppose the device always produces quantum states independent and identically (i.i.d. assumption). Denote the projector to the target subspace $\mathcal{V}$ as $\mathcal{P}$, and denote the expected quantum states produced by $D$ as $\rho$. We can use the fidelity between $\rho$ and $\mathcal{P}$ defined as 
\begin{equation}
\begin{aligned}
F(\mathcal{P},\rho)&:=\Tr{\mathcal{P}\rho}
\end{aligned}
\end{equation}
to characterize the overlap between a quantum state and the target subspace. Then, the verification problem can be described by the following hypothesis testing problem \cite{Huang2024Certifying}.
\begin{enumerate}
\item[1.] $D$ is ``good": The quantum states $\rho$ generated by $D$ satisfies $F(\mathcal{P},\rho)\geq 1-\tau\epsilon$.
\item[2.] $D$ is ``bad": The quantum states $\rho$ generated by $D$ satisfies $F(\mathcal{P},\rho)\leq 1-\epsilon$.
\end{enumerate}

Notice that a gap characterized by $\tau<1$ is introduced to tolerate imperfect verification strategy designing and measurement result fluctuation. Ideally, one would like to minimize this gap, approaching $\tau\rightarrow1$. However, while measurement result fluctuations can be reduced to any small positive constant with sufficient samples, the impact of an imperfect verification strategy design cannot be disregarded. We will later provide an upper bound for $\tau$, taking these considerations into account.

To achieve this task, one can adopt a series of 2-outcome positive-operator-valued measurements (POVMs) $\{E_l,\id-E_l\}_l$. If the measurement outcome corresponds to $E_l$, it is considered to have passed the test; otherwise, it is deemed to have failed. These POVMs should be designed such that any state $\rho$ within the subspace $\mathcal{V}$ will always pass the test, i.e., $\Tr{E_l\rho}=1, \forall \rho\in \mathcal{V}$. In each verification round, we choose a POVM $\{E_l,\id-E_l\}$ from a probability distribution $p_l$ with $\sum_l p_l=1$. Define
\begin{equation}
\begin{aligned}
    \Omega=\sum_l p_l E_l
\end{aligned}
\end{equation}
as the verification operator. By construction, one thus has $\Tr{\Omega\sigma}=\sum_l p_l\Tr{E_l\sigma}=1$ for any $\sigma\in \mathcal{V}$. The following fact gives the spectral decomposition of $\Omega$.
\begin{fact}\label{Prop:omega}
The verification operator $\Omega$ for the subspace $\mathcal{V}$ shows the following spectral decomposition
\begin{equation}\label{eq:decom}
\begin{aligned}
  \Omega=\mathcal{P}+\sum_{j=d_{\mathcal{V}}+1}^d \lambda_j\ketbra{\psi_j}{\psi_j}
\end{aligned}
\end{equation}
where $\mathcal{P}$ is the projector of $\mathcal{V}$, and $\{\ket{\psi_j}\}_{k=d_{\mathcal{V}}+1}^{d}$ is a set of basis for the orthogonal complement $\mathcal{V}^{\bot}$. $d$ and $d_{\mathcal{V}}$ are the dimensions of the Hilbert space and the subspace $\mathcal{V}$, respectively. The eigenvalues are arranged in descending order.
\end{fact}
Note that the largest eigenvalue of $\Omega$ is $1$. We denote the spectral gap between the largest eigenvalue and the second largest one as $\Delta_{\min}(\Omega)=1-\lambda_{d_{\mathcal{V}}+1}$, and the spectral gap between the largest eigenvalue and the smallest one as $\Delta_{\max}(\Omega)=1-\lambda_{d}$, and it is clear that $\Delta_{\max}(\Omega)\geq \Delta_{\min}(\Omega)$. 

Define $r:=\frac{\Delta_{\min}(\Omega)}{\tau\Delta_{\max}(\Omega)}$ and $\tilde{\epsilon}:=\Delta_{\min}(\Omega)\epsilon$ for short. Denote 
\begin{equation}\label{eq:defp0}
p_0=\frac{\ln{r}}{\ln{r}+\ln{\frac{1-\tilde{\epsilon}/r}{1-\tilde{\epsilon}}}}
\end{equation}
as the solution to the equation $D\left[p_0\|1-\tilde{\epsilon}\right]=D\left[p_0\|1-\tilde{\epsilon}/r\right]$, where $D[p\|q]:= p\log{\frac{p}{q}}+(1-p)\log{\frac{1-p}{1-q}}$ stands for the Kullback–Leibler divergence for Bernoulli distribution.

Suppose we adopt the verification tests for $N$ times, and $N_{\text{pass}}$ of them pass. To assess the performance of the quantum device $D$, we evaluate whether $N_{\text{pass}}>p_0 N$. If this condition is met, we categorize $D$ as ``good"; otherwise, we categorize it as ``bad." The subsequent theorem shows the number of required verification tests for this assessment strategy.

\begin{theorem}\label{theo:verification}
For the verification problem described by parameters $\epsilon$, $\delta$, $\tau$, and the verification operator $\Omega$, if $\tau\Delta_{\max}(\Omega)<\Delta_{\min}(\Omega)$ (i.e., $r>1$), one can reach confidence $1-\delta$ with
\begin{equation}\label{eq:veri-N}
N=\frac{\ln{1/\delta}}{D\left[p_0\|1-\tilde{\epsilon}\right]}
\end{equation}
rounds of tests. More specifically, suppose we adopt the verification tests for $N$ times, and $N_{\text{pass}}$ of them pass. If the device is ``good", the probability of getting a low-pass result is upper-bounded by
\begin{equation}\label{eq:goodlow0}
\begin{aligned}
  \mathrm{Pr}\left(N_{\text{pass}}\leq p_0N\right)\leq \delta=e^{-D\left[p_0\|1-\tilde{\epsilon}/r\right]N};
\end{aligned}
\end{equation}
Conversely, if the device is ``bad", the probability of getting a high-pass result is upper-bounded by
\begin{equation}\label{eq:badhigh0}
\begin{aligned}
  \mathrm{Pr}\left(N_{\text{pass}}\geq p_0N\right)\leq \delta=e^{-D\left[p_0\|1-\tilde{\epsilon}\right]N}.
\end{aligned}
\end{equation}
Here the parameters $r$, $\tilde{\epsilon}$ and $p_0$ are defined around Eq.~\eqref{eq:defp0}.
\end{theorem}

One can prove it by bounding the passing probability, i.e., $\mathbb{E}\left[\frac{N_{\text{pass}}}{N}\right]$, with $1-\tilde{\epsilon}$ and $1-\tilde{\epsilon}/r$ and using the Chernoff-Hoeffding theorem. The detailed proof is given in Appendix \ref{app:prooftheo1}. When $\epsilon$ is sufficiently small, we can simplify Eq.~\eqref{eq:veri-N} to the first order in $\epsilon$: 
\begin{equation}
\begin{aligned}
N\sim f(r)\frac{\ln{1/\delta}}{\Delta_{\min}(\Omega)\epsilon}
\end{aligned}
\end{equation}
where $f(r):=\left(1+\frac{1}{\ln{r}}\left(1-\frac{1}{r}\right)\left(\ln{\left(\frac{1}{\ln{r}}\left(1-\frac{1}{r}\right)\right)}-1\right)\right)^{-1}$ is some constant as a function of $r$. Notably, the scaling of $N$ with respect to $\epsilon$ is $N=O\left(\frac{1}{\epsilon}\right)$, which deviates significantly from the typical scaling of $N=O\left(\frac{1}{\epsilon^2}\right)$ found in fidelity estimation \cite{Flammia2011Direct,Guhne2007toolbox}. This distinct scaling behavior is closely associated with the property of hypothesis testing that ``good" states invariably pass the test. Specifically, the verification operator can be expressed as outlined in Eq.~\eqref{eq:decom}, where the coefficient of $\mathcal{P}$ is $1$.

It is also noteworthy that our framework diverges from previous state verification frameworks \cite{Zhu2019Efficient, Zhu2019General}. Traditional frameworks require the ``good" case of the hypothesis testing problem to satisfy $F(\mathcal{P},\rho)=1$, implying that the device is judged to be ``good" only if all tests are passed. However, physical states are invariably noisy, and it is unrealistic to expect that they will consistently pass every test, rendering these frameworks less effective when the fidelity is not exceptionally high. Our framework modifies the ``good" case to be $F(\mathcal{P},\rho)\geq1-\tau\epsilon$, thereby accommodating scenarios where fidelity may not reach such ideal levels.

Based on Theorem.~\ref{theo:verification}, we can give the following proposition, which considers an estimation scenario rather than a hypothesis testing scenario. We want to derive a bound for the infidelity of the input state $\rho$ with the ratio of passed tests.

\begin{proposition}\label{prop:estimation}
Suppose we adopt the verification tests for $N$ times, and $N_{\text{pass}}$ of them pass. Denote $p=N_{\text{pass}}/N$ as the passing frequency observed in the verification, the estimated infidelity can be bounded in the interval
\begin{equation}\label{eq:estimation}
\begin{aligned}
\max\left\{\frac{1-p-\xi}{\Delta_{\max}(\Omega)},0\right\}\leq\epsilon_{\rho}\leq\min\left\{\frac{1-p+\xi}{\Delta_{\min}(\Omega)},1\right\}
\end{aligned}
\end{equation}
with approximate confidence level $1-\delta$. Here $\xi=z\left(1-\frac{\delta}{2}\right)\sqrt{\frac{p(1-p)}{N}}$ with $z(x)$ standing for the inverse of the cumulative distribution function of the standard normal distribution at $x$.
\end{proposition}

The proof is given in Appendix \ref{app:proofprop1}. The physical interpretation of Eq.~\eqref{eq:estimation} is straightforward. The term $\xi$ arises from statistical fluctuations, proportional to the square root of $N$, and approaches zero as $N\rightarrow\infty$. The factors $\frac{1}{\Delta_{\max}(\Omega)}$ and $\frac{1}{\Delta_{\min}(\Omega)}$ result from imperfect verification operator design. They will equal $1$ if $\Omega$ is ideally designed, i.e., $\lambda_{d_{\mathcal{V}}+1}=\lambda_d=0$. The remaining term, $1-p$, represents the expected infidelity. Notice that in Eq.~\eqref{eq:estimation}, the uncertainty of $\epsilon$ is proportional to $\frac{1}{\sqrt{N}}$ rather than $\frac{1}{N}$. This distinction indicates that the specific scaling of sample complexity only appears in the hypothesis testing scenario.

Given parameters $\epsilon$, $\delta$, and $\tau$, our goal to minimize the net sample complexity $N$ involves optimizing the construction of the verification operator $\Omega$. In the subsequent two sections, we will focus on developing verification operators for two typical QEC codes: stabilizer codes and QLDPC codes. The main concern is to improve the spectral gaps and reduce the number of different measurement settings, which will be introduced in the next section.

\section{Verification of stabilizer code subspace}\label{sec:stabcode}

Stabilizer codes play an essential role in the designing of quantum error correction code, primarily because the stabilizer formalism facilitates the adaptation of some classical codes into quantum codes. For an $[[n,k,d]]$ stabilizer code, the code subspace $\mathcal{V}$ is characterized by its stabilizer group $\mathbb{S}=\bk{S_i}_{i=1}^m$ where $m=n-k$. The projector onto the code subspace $\mathcal{V}$ can be expressed as 
\begin{equation}\label{eq:Pdecomall}
\begin{aligned}
  \mathcal{P}=\prod_{i=1}^m\frac{\id+S_i}{2}=\frac{1}{2^m}\sum_{S\in\mathbb{S}}S.
\end{aligned}
\end{equation}

If there is no restriction on the measurement, one can directly adopt $\{\mathcal{P}, \id-\mathcal{P}\}$ as the POVM, i.e., the projector itself, to do the verification. In this way, the verification operator is actually $\Omega=\mathcal{P}$ with the spectral gaps $\Delta_{\min}(\mathcal{P})=\Delta_{\max}(\mathcal{P})=1$. Although the performance is optimal, implementing multi-qubit measurements on real quantum devices presents substantial challenges. Therefore, our focus shifts to the verification operator constructed from $s$-local measurements with constant $s$, which can be implemented using the tensor product of at most $s$-qubit measurements. Specifically, when $s=1$, indicating the use of only single-qubit measurements, we refer to these as local measurements in subsequent discussions. The number of local measurement settings required to realize $\Omega$ is a critical parameter in the design of the verification operator.

As all the stabilizers are Pauli operators that can be measured with local measurements, a natural construction is to measure all the stabilizers $S$ in $\mathbb{S}$ with local measurements with equal probability. The verification operator is 
\begin{equation}
\begin{aligned}
\Omega_{\text{all}}=\frac{1}{2^m}\sum_{S\in\mathbb{S}}\frac{\id+S}{2}=\frac{\id+\mathcal{P}}{2}.
\end{aligned}
\end{equation}
The spectral gaps are $\Delta_{\min}(\Omega_{\text{all}})=\Delta_{\max}(\Omega_{\text{all}})=1/2$, suggesting that the number of verification tests required scales as $O\left(\frac{\ln{1/\delta}}{\epsilon}\right)$, independent of $m$. Nevertheless, each $S\in\mathcal{S}$ corresponds to a local measurement setting. The number of measurement settings required scales exponentially with $m$, potentially complicating experimental implementation.

To reduce the number of measurement settings, one can turn to measure all the stabilizer generators $S_i$ of the subspace with equal probability. It needs at most $m$ different measurement settings. The verification operator becomes
\begin{equation}
\begin{aligned}
\Omega_{\text{gen}}=\frac{1}{m}\sum_{i=1}^m \frac{\id+S_i}{2}=\frac{1}{2}\left(\id+\frac1{m}\sum_{i=1}^m S_i\right).
\end{aligned}
\end{equation}
The spectral gaps are given by $\Delta_{\min}(\Omega_{\text{gen}})=1/m$ and $\Delta_{\max}(\Omega_{\text{gen}})=1$, which will be proven in Appendix \ref{app:stabilizerdiscussion}. Due to the non-constant spectral gap, the number of verification tests required scales as $O\left(\frac{\ln{1/\delta}}{\epsilon}m\right)$, which hinders the performance of the strategy for a large $m$.

We further introduce a new construction of $\Omega$ to improve the spectral gap and reduce the number of measurement settings. Our method is inspired by the chromatic property of the graph state \cite{Toth2005Detecting, Toth2005stabilizer,Zhou2019structure,zhao2020constructing}. We start by defining the concept of bit-wise commutativity of stabilizers and the corresponding graph.
\begin{definition}\label{def:satbGraph}
Two stabilizers on the same system, $S$ and $S'$, bit-wise commute with each other iff their reduced operators on any qubit $l$ commute, $[S(l),S'(l)]=0$.

The bit-wise commutativity graph $G_{\mathcal{S}}$ associated with a stabilizer generator set $\mathcal{S}:=\{S_1,S_2,\cdots,S_m\}$ is defined as follows: The vertex set $V=\{v_i\}_{i=1}^m$ represents the stabilizer generators, and the edge set $E=\{e_{ij}\}$ contains an edge $e_{ij}$ iff $S_i$ and $S_j$ do not bit-wise commute.
\end{definition}
We remark that if one focuses on the typical class of stabilizer state, the graph states, the introduced bit-wise commutativity graph becomes the corresponding graph of graph states as one considers the generator associated with each vertex \cite{Hein2006Graph}.

If $S_i$ and $S_j$ bit-wise commute, there should be $S_i(l)=\id$ or $S_j(l)=\id$ or $S_i(l)\cdot S_j(l)=\id$. As a result, one can measure them and their multiplication $S_iS_j$ simultaneously. Take an independent set $\mathcal{I}$ of vertices in $G_{\mathcal{S}}$, any two stabilizers corresponding to the vertices in $\mathcal{I}$ should bit-wise commute, i.e., $S_i$ and $S_j$ bit-wise commute if $v_i,v_j\in\mathcal{I}$. Now, construct a Pauli operator $M_{\mathcal{I}}$: $M_{\mathcal{I}}(l)=S_i(l)$ if there exists $v_i\in \mathcal{I}$ such that $S_i(l)\neq\id$; otherwise, $M_{\mathcal{I}}(l)=\id$. One can find that any $S_i$ corresponding to $v_i\in\mathcal{{I}}$ or their multiplication is a reduced operator of $M_{\mathcal{I}}$. Therefore, all of the stabilizers whose corresponding vertices are in $\mathcal{I}$ and their multiplications can be measured simultaneously with just a single measurement setting $M_{\mathcal{I}}$. Here, a measurement setting $M_{\mathcal{I}}$ means to assign the measurement basis according to the single qubit Pauli operators of $M_{\mathcal{I}}$. Furthermore, the POVM $\{\mathcal{P}_{\mathcal{I}},\id-\mathcal{P}_{\mathcal{I}}\}$ where 
\begin{equation}\label{eq:PIu}
\begin{aligned}
\mathcal{P}_{\mathcal{I}}=\prod_{v_i\in \mathcal{I}}\frac{\id+S_i}{2}
\end{aligned}
\end{equation}
can be realized by local measurement setting $M_{\mathcal{I}}$ and post-processing.

\begin{figure}[htbp]
    \centering
    \subfigure[]{\includegraphics[width=0.45\linewidth]{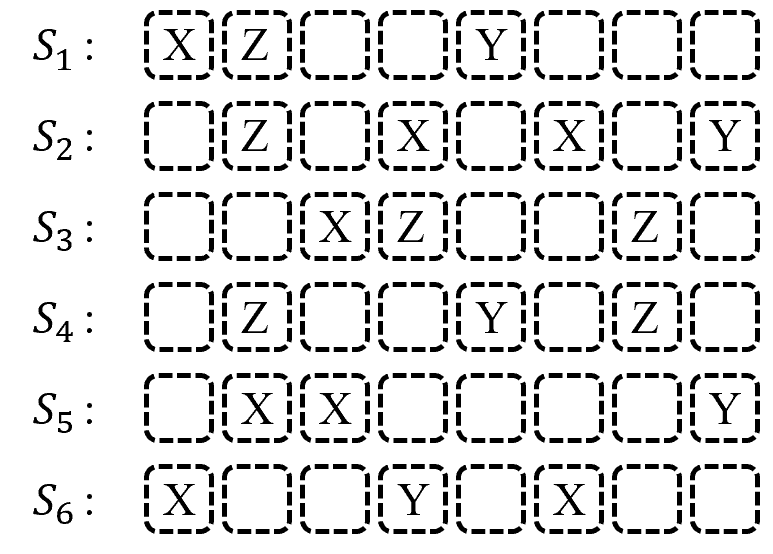}}
    \subfigure[]{\includegraphics[width=0.35\linewidth]{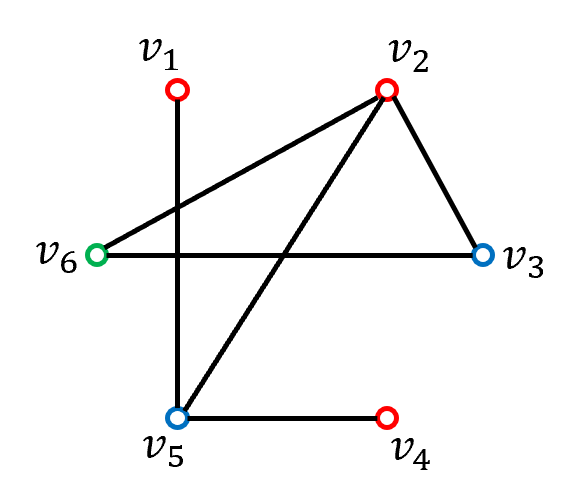}}
    \subfigure[]{\includegraphics[width=0.45\linewidth]{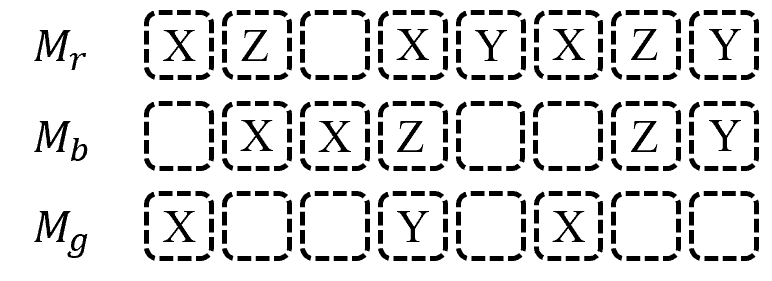}}
    \caption{\textbf{Construction of $\Omega_{\text{chr}}(\mathcal{S})$:} (a) an example set of stabilizer generators; (b) the corresponding bit-wise commutativity graph of stabilizers in (a); (c) measurement settings according to the chromatic method in (b).}
    \label{fig:graphverification}
\end{figure}

An optimal coloring method $\{\mathcal{I}_u\}_{u=1}^{\chi(G_{\mathcal{S}})}$ of $G_{\mathcal{S}}$ divides vertices into a minimal number of independent sets, i.e., $V=\bigcup_{u=1}^{\chi(G_{\mathcal{S}})} \mathcal{I}_u$ and $\mathcal{I}_u\cap \mathcal{I}_v=\emptyset$ for $u\neq v$. $\chi(G_{\mathcal{S}})$ stands for the chromatic number of $G_{\mathcal{S}}$. We can construct a verification operator as
\begin{equation}\label{Eq:OmegaColor}
\begin{aligned}
\Omega_{\text{chr}}(\mathcal{S})=\frac{1}{\chi(G_{\mathcal{S}})}\sum_{u=1}^{\chi(G_{\mathcal{S}})}\mathcal{P}_{\mathcal{I}_u}=\frac{1}{\chi(G_{\mathcal{S}})}\sum_{u=1}^{\chi(G_{\mathcal{S}})}\prod_{v_i\in \mathcal{I}_u}\frac{\id+S_i}{2}.
\end{aligned}
\end{equation}
It is clear that $\Omega_{\text{chr}}(\mathcal{S})$ is a legal verification operator and can be measured with $\chi(G_{\mathcal{S}})$ measurement settings. The spectral gaps show $\Delta_{\min}(\Omega_{\text{chr}}(\mathcal{S}))=1/\chi(G_{\mathcal{S}})$ and $\Delta_{\max}(\Omega_{\text{chr}}(\mathcal{S}))=1$, which will be proved in the Appendix \ref{app:stabilizerdiscussion}. 

Notice that the construction of $\Omega_{\text{chr}}(\mathcal{S})$ is uniquely determined by the set of stabilizer generators $\mathcal{S}$ rather than the stabilizer group $\mathbb{S}$. As there are different choices of $\mathcal{S}$ given $\mathbb{S}$, one can further optimize the construction of $\Omega_{\text{chr}}(\mathcal{S})$ by choosing $\mathcal{S}$ such that the chromatic number $\chi(G_{\mathcal{S}})$ is minimal. Denote the optimal set of stabilizer generators as $\mathcal{S}^{*}$ and the corresponding chromatic number as $\chi^{*}$. The optimal choice of $\Omega_{\text{chr}}(\mathcal{S})$ is thus
\begin{equation}\label{Eq:OmegaOpt}
\begin{aligned}
\Omega_{\text{chr}}^{*}=\frac{1}{\chi^{*}}\sum_{u=1}^{\chi^{*}}\mathcal{P}_{\mathcal{I}_u},
\end{aligned}
\end{equation}
with $\chi^{*}$ measurement settings and spectral gaps $\Delta_{\min}(\Omega_{\text{chr}}^{*})=1/\chi^{*}$, $\Delta_{\max}(\Omega_{\text{chr}}^{*})=1$.

Compared with $\Omega_{\text{gen}}$, the verification operator $\Omega_{\text{chr}}^{*}$ reduces the number of measurement setting required from $m$ to $\chi^{*}$, and improve the spectral gap $\Delta_{\min}(\Omega)$ from $1/m$ to $1/\chi^{*}$. For large-scale QEC codes, the optimal chromatic number $\chi^{*}$ usually scales much slower than $m$ and sometimes may even be a constant. For example, CSS codes \cite{Calderbank1996CSScode,Steane1996CSScode}, including various practical QEC codes such as surface codes \cite{Bravyi1998Quantum}, toric codes \cite{KITAEV2003Fault}, and hypergraph product codes \cite{Tillich2014HGPcode}, always have optimal chromatic number $\chi^{*}=2$, showing a significant performance improvement \cite{Toth2005Detecting, Zhou2019structure}.

\section{Verification of QLDPC code subspace}
Quantum low-density parity-check (QLDPC) codes with a constant encoding rate have been demonstrated to reduce the overhead of fault-tolerant quantum computation to a constant level \cite{Gottesman2014Faulttolerant}. This significant reduction positions them as a promising candidate for fault-tolerant quantum computing. For an $[[n,k,d]]$ QLDPC code, the code subspace $\mathcal{V}$ is uniquely determined by $m=n-k$ pairwise commuting projectors $\Pi=\{\pi_i\}_{i=1}^m$, such that any quantum state $\ket{\psi}\in\mathcal{V}$ should satisfies $\pi_i\ket{\psi}=\ket{\psi}$ for all $i\in [m]$. The projector onto the QLDPC code subspace $\mathcal{V}$ is given by $\mathcal{P}=\prod_{i=1}\pi_i$, analogous to stabilizer codes. We denote a subsystem $A=\text{supp}(O)$ as the non-trivial support of any operator $O$ iff $O$ can be decomposed as $O=O_A \otimes \id_{\bar{A}}$, and one cannot find a smaller subsystem $A'\subset A$ such that $O=O_{A'}\otimes \id_{\bar{A'}}$. The sparsity of $\Pi$ is defined as $s:=\max_i \left|\text{supp}(\pi_i)\right|$. For QLDPC codes, $s$ should be some constant independent of $n$.

We can also define a support graph, analog to the previously introduced Def.~\ref{def:satbGraph} for stabilizer states.
\begin{definition}
The support graph $G_{\Pi}$ associated with the set of projectors $\Pi=\{\pi_i\}_{i=1}^m$ is defined as follows: The vertex set $V=\{v_i\}_{i=1}^m$ represents the projectors, and the edge set $E=\{e_{ij}\}$ contains an edge $e_{ij}$ iff $\text{supp}(\pi_i)\cap\text{supp}(\pi_j)\neq\emptyset$.
\end{definition}
The QLDPC condition requires that each qubit appears in no more than $s$ supports of the projectors. Combining with the fact that each projector acts non-trivially on at most $s$ qubits, the support of any projector intersects with at most $s^2$ other supports, which implies that the degree of $G_{\Pi}$ is at most $s^2$. Consequently, the support graph $G_{\Pi}$ of any QLDPC code with sparsity $s$ can be colored with no more than $s^2$ colors.

If $s$-local measurements are permitted, verification operators similar to those discussed in the previous section can be constructed, as each POVM $\{\pi_i,\id-\pi_i\}$ can be realized with a single $s$-local measurement. To be specific, by measuring each projector with equal probability, we can construct 
\begin{equation}\label{eq:LDPCgens}
\begin{aligned}
\Omega_{\text{gen}}^{(s)}=\frac{1}{m}\sum_{i=1}^m \pi_i
\end{aligned}
\end{equation}
with spectral gaps $\Delta_{\min}(\Omega_{\text{gen}}^{(s)})=1/m$ and $\Delta_{\max}(\Omega_{\text{gen}}^{(s)})=1$. Here, the superscript $^{(s)}$ implies $s$-local measurements are used. It is not hard to find that $\Omega_{\text{gen}}^{(s)}$ can be realized with $m$ different $s$-local measurement settings. Moreover, with a coloring method of the support graph $G_{\Pi}$, we can construct 
\begin{equation}\label{eq:LDPCchrs}
\begin{aligned}
\Omega_{\text{chr}}^{(s)}=\frac{1}{\chi(G_{\Pi})}\sum_{u=1}^{\chi(G_{\Pi})}\prod_{v_i\in\mathcal{I}_u}\pi_i
\end{aligned}
\end{equation}
with spectral gaps $\Delta_{\min}(\Omega_{\text{chr}}^{(s)})=1/\chi(G_{\Pi})\geq 1/s^2$ and $\Delta_{\max}(\Omega_{\text{chr}}^{(s)})=1$. Each $\prod_{v_i\in\mathcal{I}_u}\pi_i$ can be measured with a single measurement setting since different $\pi_i$ in the same independent set act on different supports, which indicates that $\Omega_{\text{chr}}^{(s)}$ can be realized with $\chi(G_{\Pi})\leq s^2$ different measurement settings.

For QLDPC stabilizer codes \cite{Gottesman2014Faulttolerant, Panteleev2022GoodLDPC, xu2024constant}, since all the $s$-local projectors are $+1$ projectors of Pauli operators, they can be measured with $1$-local measurements and post-processing. However, in general, $\{\pi_i,\id-\pi_i\}$ cannot be directly implemented by a single local measurement setting and post-processing. The implementation of $s$-local measurements often faces practical challenges due to the geometric limitations of the qubits. Consequently, there is a continued need to develop methodologies that utilize $1$-local measurements.

To address this problem, we decompose the projectors $\pi_i$ into Pauli measurements that are feasible to implement. We also present the associated verification operator along with its spectral information. Further details and proofs are provided in Appendix \ref{app:qldpcdiscussion}.

Let $T_i=2\pi_i-\id$ as the measurement operator of $\{\pi_i,\id-\pi_i\}$. One can decompose $T_i=\sum_{P_j\in \mathbb{P}_n} t_{ij}P_j$ where $t_{ij}=\Tr{T_i P_j}$ are positive numbers and $\mathbb{P}_n$ denotes the $n$-qubit Pauli group with only $\pm 1$ phase. All the $P_j$ with nonzero coefficient $t_{ij}$ should act on the sub-supports of $A_i$, which is the support of $T_i$ (or $\pi_i$). Denote each set of Pauli operators with nonzero $t_{ij}$ with $\mathcal{T}_i$, and let $a_i=\sum_{j}t_{ij}$, $p_{ij}=t_{ij}/a_i$ for simplicity. Therefore, $\{p_{ij}\}_j$ can be regarded as a probability distribution. Since $t_{ij}\in[0,1]$ and $\frac{1}{2^n}\Tr{T_i}=\sum_{j}t_{ij}^2=1$, we can obtain $1\leq a_i\leq \sqrt{|\mathcal{T}_i|}\leq 2^s$. The decomposition can be written as $T_i=a_i\sum_{P_{i,j}\in \mathcal{T}_i} p_{ij}P_{i,j}$. Based on the decomposition, the verification operator can be constructed as 
\begin{equation}\label{eq:LDPCgen1}
\begin{aligned}
\Omega_{\text{gen}}^{(1)}&=\frac{1}{m}\sum_i\left(\frac{1}{a_i}\pi_i+\left(1-\frac{1}{a_i}\right)\frac{\id}{2}\right)\\
&=\frac{1}{m}\sum_i\left(\frac{1}{2}\sum_{P_{i,j}\in \mathcal{T}_i} p_{ij}P_{i,j}+\frac{\id}{2}\right)\\
&=\frac{1}{m}\sum_i\sum_{P_{i,j}\in \mathcal{T}_i}p_{ij}\frac{\id+P_{i,j}}{2}.
\end{aligned}
\end{equation}

The corresponding spectral gaps are given by 
\begin{equation}\label{}
\begin{aligned}
\Delta_{\min}(\Omega_{\text{gen}}^{(1)})=\frac{1}{m\cdot\max_i a_i},\ \Delta_{\max}(\Omega_{\text{gen}}^{(1)})=\frac{1}{m}\sum_i\frac{1}{a_i}.
\end{aligned}
\end{equation}
Since $1\leq a_i\leq \sqrt{|\mathcal{T}_i|}\leq 2^s$, the spectral gaps $\Delta_{\min}(\Omega_{\text{gen}}^{(1)})\geq \frac{1}{2^s m}$ and $\Delta_{\max}(\Omega_{\text{gen}}^{(1)})\geq \frac{1}{2^s}$. Meanwhile, since $\{p_{ij}\}_j$ can be regarded as a probability distribution, such a verification operator can be realized by measuring $P_{i,j}$ with probability $p_{ij}$, totally no more than $\sum_i |\mathcal{T}_i|\leq m\cdot 4^s$ measurement settings. 

Importantly, note that $\Omega_{\text{gen}}^{(1)}$ does not have the maximum eigenvalue $+1$ anymore, as shown in Appendix \ref{app:qldpcdiscussion}, which implies that perfect QLDPC code states might not pass the test. This indicates that the analysis in Theorem \ref{theo:verification} no longer applies. Therefore, we give the following refined treatment, which would be of independent interest.

\begin{proposition}\label{prop:imperfectomega}
For the verification problem described by parameters $\epsilon$, $\delta$, $\tau$, and the verification operator $\Omega$, suppose the verification operator $\Omega$ has the maximum eigenvalue $\lambda<1$, i.e., $\Omega=\lambda \mathcal{P}+\sum_{j=d_{\mathcal{V}}+1}^d \lambda_j\ketbra{\psi_j}{\psi_j}$. If $\tau\Delta_{\max}(\Omega)<\Delta_{\min}(\Omega)$, one can reach confidence $1-\delta$ with the judgement passing frequency $p'_0$ given in Eq.\eqref{eq:pprime0} as a function of $\lambda$, $\tau$, $\epsilon$ and spectral gaps of $\Omega$, and 
\begin{equation}\label{eq:veri-Nprime}
N\sim \frac{8\lambda(1-\lambda)\ln{1/\delta}}{\left(\Delta_{\min}(\Omega)-\Delta_{\max}(\Omega)\tau\right)^2 \epsilon^2}.
\end{equation}
rounds of tests. In other words, if $N_{\text{pass}}\geq p'_0 N$, we categorize $D$ as ``good"; otherwise, we categorize it as ``bad."
\end{proposition}

The proof is analogous to that of Theorem \ref{theo:verification} and given in Appendix \ref{app:proofprop2}. If one regards $\lambda$ as a constant for general case, $N=O\left(\frac{\ln{1/\delta}}{\Delta_{\min}^2(\Omega)\epsilon^2}\right)$ scales differently as that in Theorem \ref{theo:verification}. This difference comes from the condition that the maximum eigenvalue of $\Omega$ is no longer $1$, i.e., the states in the target subspace may also fail to pass the test. For $\Omega_{\text{gen}}^{(1)}$, the maximum eigenvalue, denoted as $\lambda$, is calculated as $\lambda=\frac{1}{2}\left(1+\frac{1}{m}\sum_i \frac{1}{a_i}\right)$, where it is trivially bounded within the interval $[\frac{1}{2},1]$. Consequently, $\lambda$ can be considered a constant, and the number tests required scales $N=O\left(\frac{m^2\ln{1/\delta}}{\epsilon^2}\right)$.

One may expect the enhancement construction of $\Omega_{\text{chr}}^{(1)}$, analogous to that in Eq.\eqref{eq:LDPCchrs}, which can be implemented with $1$-local measurement settings. However, such construction typically results in exponentially small eigenvalues, making it ineffective for verification tasks. Detailed discussion is given in Appendix \ref{app:qldpcdiscussion}.

\section{Fidelity estimation of logical states}

A common challenge in quantum information processing and quantum computing involves estimating the fidelity between a quantum state $\rho$ and a target state $\ket{\psi}$. Two prevalent methods, quantum state verification \cite{Sam2018Optimal, Zhu2019Efficient, Zhu2019General, yu2019optimal}, and direct fidelity estimation \cite{Flammia2011Direct, Guhne2007toolbox, Seshadri2024versatile}, exhibit specific constraints respectively. Quantum state verification, while relatively efficient, typically relies on the stabilizer structure of the target state $\ket{\psi}$, making it difficult to develop efficient protocols for general magic states \cite{Zhu2019Hypergraph}. Conversely, direct fidelity estimation offers broader applicability across various quantum states but demands substantial resources, especially for magic states. The sample complexity of this method scales exponentially with the number of qubits $n$, specifically $O\left(\frac{1}{\epsilon^2\delta}+\frac{2^n}{\epsilon^2}\ln{1/\delta}\right)$, rendering it unsuitable for large systems. By integrating quantum subspace verification with direct fidelity estimation, we propose a composite protocol for estimating the fidelity between a state $\rho$ and a QEC code state $\ket{\psi}$ as follows, where $\ket{\psi}$ could be any magic state in the code space.

One first adopts the subspace verification protocol to verify the fidelity between the given state $\rho$ and the code subspace $\mathcal{V}$: $\Tr{\mathcal{P}_{\mathcal{V}}\rho}\geq 1-\epsilon_1$ with confidence $1-\delta_1$. Specifically, based on Eq.~\eqref{eq:badhigh0} in Theorem \ref{theo:verification}, one has
\begin{equation}
\mathrm{Pr}\left(N_{\text{pass}}\geq p_0N_1|\Tr{\mathcal{P}_{\mathcal{V}}\rho}\leq 1-\epsilon_1\right)\leq\delta_1.
\end{equation}

Following this, we creatively implement the direct fidelity estimation scheme on the logical qubits, specifically aiming to estimate the ``logical fidelity" defined as
\begin{equation}
\bar{F}:=\frac{1}{2^k}\sum_{\bar{P}}\Tr{\bar{P}\ketbra{\psi}{\psi}}\cdot\Tr{\bar{P}\rho},
\end{equation}
where $\bar{P}$ denotes all the logical Pauli operators on the $k$-dimensional logical subspace. Following direct fidelity estimation, we construct an unbiased estimator $Y$ such that $\mathbb{E}[Y]=\bar{F}$, and by Chebyshev’s inequality \cite{Flammia2011Direct}, 
\begin{equation}
\mathrm{Pr}\left(Y\geq \mathbb{E}[Y]+\epsilon_2\right)\leq\delta_2.
\end{equation}
Detailed illustrations are given in Appendix \ref{app:prooftheo2}.

The previous two steps can return an estimation of the target fidelity, and the following theorem shows the net sample complexity of the composite protocol for the fidelity estimation.
\begin{theorem}\label{theo:QSV+DFE} 
Suppose one applies the composite protocol to estimate the fidelity between the prepared state $\rho$ with the target state $\ket{\psi}$ in the code subspace, namely $\Tr{\ketbra{\psi}{\psi}\rho}$.
Assume that in the initial step, the input states $\rho$ pass the subspace verification using $N_1$ samples with verification operator $\Omega$, i.e., $N_{\text{pass}} \geq p_0 N_1$, under parameters $\epsilon$ and $\delta$; In the subsequent step, one obtains the estimator $Y$ using $N_2$ samples under the same parameters. 
It can be concluded that the fidelity $\Tr{\ketbra{\psi}{\psi}\rho} \geq Y - 2\epsilon$ with a confidence level of $1 - 2\delta$. The overall sample complexity is given by $N = N_1+N_2=O\left(\frac{\ln{1/\delta}}{\Delta_{\min}(\Omega) \epsilon}\right) + O\left(\frac{1}{\epsilon^2 \delta} + \frac{2^k \ln{1/\delta}}{\epsilon^2}\right)$.
\end{theorem}

The proof is given in Appendix \ref{app:prooftheo2}. As shown in Section.~\ref{sec:stabcode}, $\Omega$ can be constructed such that $\Delta_{\min}(\Omega)\geq \frac{1}{m}$ for stabilizer codes with $m=n-k$. As a result, in most cases, the second term is the leading term. The sample complexity of our protocol is reduced by $2^m$ times compared to that of direct fidelity estimation, making it much more efficient than directly applying that method. For CSS codes and QLDPC stabilizer codes, $\Delta_{\min}(\Omega)$ could be a constant, which implies that the sample complexity is only related to the logical qubit number, independent of the physical qubit number.

\section{Conclusion and Discussion}
In this work, we expand the concept of quantum state verification by proposing a framework for quantum subspace verification. This framework addresses the challenge of benchmarking QEC codes in physical systems. Our theoretical analysis reveals that the spectral gaps of the verification operator play a critical role in determining the number of measurements required. We construct various verification operators for two widely used types of QEC codes: stabilizer codes and QLDPC codes. The verification process can be efficiently realized with $O(\frac{m\ln{1/\delta}}{\epsilon})$ and $O(\frac{m^2\ln{1/\delta}}{\epsilon^2})$ local measurements for generic stabilizer codes and QLDPC codes, respectively. Remarkably, for specific codes such as CSS codes and QLDPC stabilizer codes, the number of required local measurements does not depend on $m$. These developments enable efficient and effective benchmarking of QEC codes in practical fault-tolerant quantum computing. Besides, we introduce a new composite protocol for fidelity estimation of QEC code states that integrates quantum subspace verification with direct fidelity estimation. This hybrid approach enables efficient fidelity estimation between $\rho$ and any QEC code state $\ket{\psi}$.

Subsequent efforts may focus on extending our results to encompass a broader array of QEC codes. A promising direction is to adapt our framework for continuous-variable quantum systems, specifically to verify QEC codes such as Gottesman-Kitaev-Preskill (GKP) codes \cite{Gottesman2001GKPcode,Liu2021CV} within these systems. Additionally, exploring the integration of our subspace verification approach with approximate QEC codes \cite{Leung1997Approximate, Hayden2020Approximate} presents an intriguing possibility. These codes are pivotal in fault-tolerant quantum computing, benefiting from their capacity to support universal transversal gates. Furthermore, it remains crucial to determine whether almost all QEC subspaces that exhibit good performance can be efficiently verified using local measurements, as demonstrated for quantum state verification \cite{Huang2024Certifying}.

To broaden the applicability of our framework, it is also crucial to extend it to adversarial scenarios, which would relax the i.i.d. assumption of the input states. Techniques employed in the framework of quantum state verification might be used to address the problem \cite{Zhu2019Efficient}. In addition, the extension to the quantum process is also interesting \cite{liu2019efficient,zhu2019efficientgate,Zeng2020gate}. Finally, the current work mainly discusses the combination of quantum subspace verification and direct fidelity estimation. It is thus intriguing to explore the further application of quantum subspace verification on general quantum state learning problems \cite{Anshu2024Survey,gebhart2023learning}.

While finalizing our manuscript, we note that a contemporaneous work also introduces the concept of subspace verification \cite{zheng2024efficient}. Compared with theirs, our framework utilizes a more general hypothesis testing scheme better suited for noisy physical states, rather than merely extending the existing quantum state verification frameworks to include subspace. Additionally, our research focuses on a broader range of stabilizer codes, including generic stabilizer and QLDPC codes, rather than specific codes. Furthermore, we introduce a protocol that leverages quantum subspace verification for quantum state learning, illustrating a practical application of our theoretical advancements.

\section*{Acknowledgments}
We thank Guoding Liu, Zihao Li, Zhenhuan Liu, Yuxuan Yan, Huangjun Zhu, Xingjian Zhang for their useful discussions. J.~C and X.~M acknowledge the support of the National Natural Science Foundation of China Grant No.~12174216 and the Innovation Program for Quantum Science and Technology Grant No.~2021ZD0300804 and No.~2021ZD0300702.
Q.~Z acknowledges the support of the HKU Seed Fund for Basic
Research for New Staff via Project 2201100596, Guangdong Natural Science Fund via Project 2023A1515012185, National Natural Science Foundation of China (NSFC) via Project No. 12305030 and No. 12347104, Hong Kong Research Grant Council (RGC) via No. 27300823, N\_HKU718/23, and R6010-23, Guangdong Provincial Quantum Science Strategic Initiative GDZX2200001.
Y.~Z acknowledges the support of National Natural Science Foundation of China (NSFC) Grant No.12205048, Innovation Program for Quantum Science and Technology 2021ZD0302000, the start-up funding of Fudan University, and the CPS-Huawei MindSpore Fellowship.


\bibliography{BibSubspace}

\begin{thebibliography}{46}%
\makeatletter
\providecommand \@ifxundefined [1]{%
 \@ifx{#1\undefined}
}%
\providecommand \@ifnum [1]{%
 \ifnum #1\expandafter \@firstoftwo
 \else \expandafter \@secondoftwo
 \fi
}%
\providecommand \@ifx [1]{%
 \ifx #1\expandafter \@firstoftwo
 \else \expandafter \@secondoftwo
 \fi
}%
\providecommand \natexlab [1]{#1}%
\providecommand \enquote  [1]{``#1''}%
\providecommand \bibnamefont  [1]{#1}%
\providecommand \bibfnamefont [1]{#1}%
\providecommand \citenamefont [1]{#1}%
\providecommand \href@noop [0]{\@secondoftwo}%
\providecommand \href [0]{\begingroup \@sanitize@url \@href}%
\providecommand \@href[1]{\@@startlink{#1}\@@href}%
\providecommand \@@href[1]{\endgroup#1\@@endlink}%
\providecommand \@sanitize@url [0]{\catcode `\\12\catcode `\$12\catcode
  `\&12\catcode `\#12\catcode `\^12\catcode `\_12\catcode `\%12\relax}%
\providecommand \@@startlink[1]{}%
\providecommand \@@endlink[0]{}%
\providecommand \url  [0]{\begingroup\@sanitize@url \@url }%
\providecommand \@url [1]{\endgroup\@href {#1}{\urlprefix }}%
\providecommand \urlprefix  [0]{URL }%
\providecommand \Eprint [0]{\href }%
\providecommand \doibase [0]{https://doi.org/}%
\providecommand \selectlanguage [0]{\@gobble}%
\providecommand \bibinfo  [0]{\@secondoftwo}%
\providecommand \bibfield  [0]{\@secondoftwo}%
\providecommand \translation [1]{[#1]}%
\providecommand \BibitemOpen [0]{}%
\providecommand \bibitemStop [0]{}%
\providecommand \bibitemNoStop [0]{.\EOS\space}%
\providecommand \EOS [0]{\spacefactor3000\relax}%
\providecommand \BibitemShut  [1]{\csname bibitem#1\endcsname}%
\let\auto@bib@innerbib\@empty
\bibitem [{\citenamefont {Kitaev}(2003{\natexlab{a}})}]{Kitaev2003Toric}%
  \BibitemOpen
  \bibfield  {author} {\bibinfo {author} {\bibfnamefont {A.}~\bibnamefont
  {Kitaev}},\ }\bibfield  {title} {\bibinfo {title} {Fault-tolerant quantum
  computation by anyons},\ }\href
  {https://doi.org/https://doi.org/10.1016/S0003-4916(02)00018-0} {\bibfield
  {journal} {\bibinfo  {journal} {Annals of Physics}\ }\textbf {\bibinfo
  {volume} {303}},\ \bibinfo {pages} {2} (\bibinfo {year}
  {2003}{\natexlab{a}})}\BibitemShut {NoStop}%
\bibitem [{\citenamefont {Bravyi}\ and\ \citenamefont
  {Kitaev}(1998)}]{Bravyi1998Quantum}%
  \BibitemOpen
  \bibfield  {author} {\bibinfo {author} {\bibfnamefont {S.~B.}\ \bibnamefont
  {Bravyi}}\ and\ \bibinfo {author} {\bibfnamefont {A.~Y.}\ \bibnamefont
  {Kitaev}},\ }\href@noop {} {\bibinfo {title} {Quantum codes on a lattice with
  boundary}} (\bibinfo {year} {1998}),\ \Eprint
  {https://arxiv.org/abs/quant-ph/9811052} {arXiv:quant-ph/9811052}
  \BibitemShut {NoStop}%
\bibitem [{\citenamefont {Acharya}\ \emph {et~al.}(2024)\citenamefont
  {Acharya}, \citenamefont {Aghababaie-Beni}, \citenamefont {Aleiner},
  \citenamefont {Andersen}, \citenamefont {Ansmann}, \citenamefont {Arute},
  \citenamefont {Arya}, \citenamefont {Asfaw}, \citenamefont {Astrakhantsev},
  \citenamefont {Atalaya} \emph {et~al.}}]{Google2024surfacecodethreshold}%
  \BibitemOpen
  \bibfield  {author} {\bibinfo {author} {\bibfnamefont {R.}~\bibnamefont
  {Acharya}}, \bibinfo {author} {\bibfnamefont {L.}~\bibnamefont
  {Aghababaie-Beni}}, \bibinfo {author} {\bibfnamefont {I.}~\bibnamefont
  {Aleiner}}, \bibinfo {author} {\bibfnamefont {T.~I.}\ \bibnamefont
  {Andersen}}, \bibinfo {author} {\bibfnamefont {M.}~\bibnamefont {Ansmann}},
  \bibinfo {author} {\bibfnamefont {F.}~\bibnamefont {Arute}}, \bibinfo
  {author} {\bibfnamefont {K.}~\bibnamefont {Arya}}, \bibinfo {author}
  {\bibfnamefont {A.}~\bibnamefont {Asfaw}}, \bibinfo {author} {\bibfnamefont
  {N.}~\bibnamefont {Astrakhantsev}}, \bibinfo {author} {\bibfnamefont
  {J.}~\bibnamefont {Atalaya}}, \emph {et~al.},\ }\href@noop {} {\bibinfo
  {title} {Quantum error correction below the surface code threshold}}
  (\bibinfo {year} {2024}),\ \Eprint {https://arxiv.org/abs/2408.13687}
  {arXiv:2408.13687 [quant-ph]} \BibitemShut {NoStop}%
\bibitem [{\citenamefont {Gottesman}(2014)}]{Gottesman2014Faulttolerant}%
  \BibitemOpen
  \bibfield  {author} {\bibinfo {author} {\bibfnamefont {D.}~\bibnamefont
  {Gottesman}},\ }\href@noop {} {\bibinfo {title} {Fault-tolerant quantum
  computation with constant overhead}} (\bibinfo {year} {2014}),\ \Eprint
  {https://arxiv.org/abs/1310.2984} {arXiv:1310.2984} \BibitemShut {NoStop}%
\bibitem [{\citenamefont {Panteleev}\ and\ \citenamefont
  {Kalachev}(2022)}]{Panteleev2022GoodLDPC}%
  \BibitemOpen
  \bibfield  {author} {\bibinfo {author} {\bibfnamefont {P.}~\bibnamefont
  {Panteleev}}\ and\ \bibinfo {author} {\bibfnamefont {G.}~\bibnamefont
  {Kalachev}},\ }\bibfield  {title} {\bibinfo {title} {Asymptotically good
  quantum and locally testable classical ldpc codes},\ }in\ \href
  {https://doi.org/10.1145/3519935.3520017} {\emph {\bibinfo {booktitle}
  {Proceedings of the 54th Annual ACM SIGACT Symposium on Theory of
  Computing}}},\ \bibinfo {series and number} {STOC 2022}\ (\bibinfo
  {publisher} {Association for Computing Machinery},\ \bibinfo {address} {New
  York, NY, USA},\ \bibinfo {year} {2022})\ p.\ \bibinfo {pages}
  {375–388}\BibitemShut {NoStop}%
\bibitem [{\citenamefont {Xu}\ \emph {et~al.}(2024)\citenamefont {Xu},
  \citenamefont {Bonilla~Ataides}, \citenamefont {Pattison}, \citenamefont
  {Raveendran}, \citenamefont {Bluvstein}, \citenamefont {Wurtz}, \citenamefont
  {Vasi{\'c}}, \citenamefont {Lukin}, \citenamefont {Jiang},\ and\
  \citenamefont {Zhou}}]{xu2024constant}%
  \BibitemOpen
  \bibfield  {author} {\bibinfo {author} {\bibfnamefont {Q.}~\bibnamefont
  {Xu}}, \bibinfo {author} {\bibfnamefont {J.~P.}\ \bibnamefont
  {Bonilla~Ataides}}, \bibinfo {author} {\bibfnamefont {C.~A.}\ \bibnamefont
  {Pattison}}, \bibinfo {author} {\bibfnamefont {N.}~\bibnamefont
  {Raveendran}}, \bibinfo {author} {\bibfnamefont {D.}~\bibnamefont
  {Bluvstein}}, \bibinfo {author} {\bibfnamefont {J.}~\bibnamefont {Wurtz}},
  \bibinfo {author} {\bibfnamefont {B.}~\bibnamefont {Vasi{\'c}}}, \bibinfo
  {author} {\bibfnamefont {M.~D.}\ \bibnamefont {Lukin}}, \bibinfo {author}
  {\bibfnamefont {L.}~\bibnamefont {Jiang}},\ and\ \bibinfo {author}
  {\bibfnamefont {H.}~\bibnamefont {Zhou}},\ }\bibfield  {title} {\bibinfo
  {title} {Constant-overhead fault-tolerant quantum computation with
  reconfigurable atom arrays},\ }\href
  {https://www.nature.com/articles/s41567-024-02479-z} {\bibfield  {journal}
  {\bibinfo  {journal} {Nature Physics}\ ,\ \bibinfo {pages} {1}} (\bibinfo
  {year} {2024})}\BibitemShut {NoStop}%
\bibitem [{\citenamefont {Gullans}\ \emph {et~al.}(2021)\citenamefont
  {Gullans}, \citenamefont {Krastanov}, \citenamefont {Huse}, \citenamefont
  {Jiang},\ and\ \citenamefont {Flammia}}]{Gullans2021RandomQEC}%
  \BibitemOpen
  \bibfield  {author} {\bibinfo {author} {\bibfnamefont {M.~J.}\ \bibnamefont
  {Gullans}}, \bibinfo {author} {\bibfnamefont {S.}~\bibnamefont {Krastanov}},
  \bibinfo {author} {\bibfnamefont {D.~A.}\ \bibnamefont {Huse}}, \bibinfo
  {author} {\bibfnamefont {L.}~\bibnamefont {Jiang}},\ and\ \bibinfo {author}
  {\bibfnamefont {S.~T.}\ \bibnamefont {Flammia}},\ }\bibfield  {title}
  {\bibinfo {title} {Quantum coding with low-depth random circuits},\ }\href
  {https://doi.org/10.1103/PhysRevX.11.031066} {\bibfield  {journal} {\bibinfo
  {journal} {Phys. Rev. X}\ }\textbf {\bibinfo {volume} {11}},\ \bibinfo
  {pages} {031066} (\bibinfo {year} {2021})}\BibitemShut {NoStop}%
\bibitem [{\citenamefont {Finsterhoelzl}\ and\ \citenamefont
  {Burkard}(2022)}]{Finsterhoelzl2023benchmarkingQEC}%
  \BibitemOpen
  \bibfield  {author} {\bibinfo {author} {\bibfnamefont {R.}~\bibnamefont
  {Finsterhoelzl}}\ and\ \bibinfo {author} {\bibfnamefont {G.}~\bibnamefont
  {Burkard}},\ }\bibfield  {title} {\bibinfo {title} {Benchmarking quantum
  error-correcting codes on quasi-linear and central-spin processors},\ }\href
  {https://doi.org/10.1088/2058-9565/aca21f} {\bibfield  {journal} {\bibinfo
  {journal} {Quantum Science and Technology}\ }\textbf {\bibinfo {volume}
  {8}},\ \bibinfo {pages} {015013} (\bibinfo {year} {2022})}\BibitemShut
  {NoStop}%
\bibitem [{\citenamefont {Cramer}\ \emph {et~al.}(2010)\citenamefont {Cramer},
  \citenamefont {Plenio}, \citenamefont {Flammia}, \citenamefont {Somma},
  \citenamefont {Gross}, \citenamefont {Bartlett}, \citenamefont
  {Landon-Cardinal}, \citenamefont {Poulin},\ and\ \citenamefont
  {Liu}}]{cramer2010efficient}%
  \BibitemOpen
  \bibfield  {author} {\bibinfo {author} {\bibfnamefont {M.}~\bibnamefont
  {Cramer}}, \bibinfo {author} {\bibfnamefont {M.~B.}\ \bibnamefont {Plenio}},
  \bibinfo {author} {\bibfnamefont {S.~T.}\ \bibnamefont {Flammia}}, \bibinfo
  {author} {\bibfnamefont {R.}~\bibnamefont {Somma}}, \bibinfo {author}
  {\bibfnamefont {D.}~\bibnamefont {Gross}}, \bibinfo {author} {\bibfnamefont
  {S.~D.}\ \bibnamefont {Bartlett}}, \bibinfo {author} {\bibfnamefont
  {O.}~\bibnamefont {Landon-Cardinal}}, \bibinfo {author} {\bibfnamefont
  {D.}~\bibnamefont {Poulin}},\ and\ \bibinfo {author} {\bibfnamefont {Y.-K.}\
  \bibnamefont {Liu}},\ }\bibfield  {title} {\bibinfo {title} {Efficient
  quantum state tomography},\ }\href
  {https://www.nature.com/articles/ncomms1147} {\bibfield  {journal} {\bibinfo
  {journal} {Nature communications}\ }\textbf {\bibinfo {volume} {1}},\
  \bibinfo {pages} {149} (\bibinfo {year} {2010})}\BibitemShut {NoStop}%
\bibitem [{\citenamefont {Gross}\ \emph {et~al.}(2010)\citenamefont {Gross},
  \citenamefont {Liu}, \citenamefont {Flammia}, \citenamefont {Becker},\ and\
  \citenamefont {Eisert}}]{gross2010quantum}%
  \BibitemOpen
  \bibfield  {author} {\bibinfo {author} {\bibfnamefont {D.}~\bibnamefont
  {Gross}}, \bibinfo {author} {\bibfnamefont {Y.-K.}\ \bibnamefont {Liu}},
  \bibinfo {author} {\bibfnamefont {S.~T.}\ \bibnamefont {Flammia}}, \bibinfo
  {author} {\bibfnamefont {S.}~\bibnamefont {Becker}},\ and\ \bibinfo {author}
  {\bibfnamefont {J.}~\bibnamefont {Eisert}},\ }\bibfield  {title} {\bibinfo
  {title} {Quantum state tomography via compressed sensing},\ }\href
  {https://doi.org/10.1103/PhysRevLett.105.150401} {\bibfield  {journal}
  {\bibinfo  {journal} {Phys. Rev. Lett.}\ }\textbf {\bibinfo {volume} {105}},\
  \bibinfo {pages} {150401} (\bibinfo {year} {2010})}\BibitemShut {NoStop}%
\bibitem [{\citenamefont {Lvovsky}\ and\ \citenamefont
  {Raymer}(2009)}]{Lvovsky2009CVtomoRMP}%
  \BibitemOpen
  \bibfield  {author} {\bibinfo {author} {\bibfnamefont {A.~I.}\ \bibnamefont
  {Lvovsky}}\ and\ \bibinfo {author} {\bibfnamefont {M.~G.}\ \bibnamefont
  {Raymer}},\ }\bibfield  {title} {\bibinfo {title} {Continuous-variable
  optical quantum-state tomography},\ }\href
  {https://doi.org/10.1103/RevModPhys.81.299} {\bibfield  {journal} {\bibinfo
  {journal} {Rev. Mod. Phys.}\ }\textbf {\bibinfo {volume} {81}},\ \bibinfo
  {pages} {299} (\bibinfo {year} {2009})}\BibitemShut {NoStop}%
\bibitem [{\citenamefont {Huang}\ \emph {et~al.}(2020)\citenamefont {Huang},
  \citenamefont {Kueng},\ and\ \citenamefont {Preskill}}]{Huang2020Predicting}%
  \BibitemOpen
  \bibfield  {author} {\bibinfo {author} {\bibfnamefont {H.-Y.}\ \bibnamefont
  {Huang}}, \bibinfo {author} {\bibfnamefont {R.}~\bibnamefont {Kueng}},\ and\
  \bibinfo {author} {\bibfnamefont {J.}~\bibnamefont {Preskill}},\ }\bibfield
  {title} {\bibinfo {title} {Predicting many properties of a quantum system
  from very few measurements},\ }\href
  {https://doi.org/10.1038/s41567-020-0932-7} {\bibfield  {journal} {\bibinfo
  {journal} {Nature Physics}\ }\textbf {\bibinfo {volume} {16}},\ \bibinfo
  {pages} {1050} (\bibinfo {year} {2020})}\BibitemShut {NoStop}%
\bibitem [{\citenamefont {Chen}\ \emph {et~al.}(2021)\citenamefont {Chen},
  \citenamefont {Yu}, \citenamefont {Zeng},\ and\ \citenamefont
  {Flammia}}]{Senrui2021Robust}%
  \BibitemOpen
  \bibfield  {author} {\bibinfo {author} {\bibfnamefont {S.}~\bibnamefont
  {Chen}}, \bibinfo {author} {\bibfnamefont {W.}~\bibnamefont {Yu}}, \bibinfo
  {author} {\bibfnamefont {P.}~\bibnamefont {Zeng}},\ and\ \bibinfo {author}
  {\bibfnamefont {S.~T.}\ \bibnamefont {Flammia}},\ }\bibfield  {title}
  {\bibinfo {title} {Robust shadow estimation},\ }\href
  {https://doi.org/10.1103/PRXQuantum.2.030348} {\bibfield  {journal} {\bibinfo
   {journal} {PRX Quantum}\ }\textbf {\bibinfo {volume} {2}},\ \bibinfo {pages}
  {030348} (\bibinfo {year} {2021})}\BibitemShut {NoStop}%
\bibitem [{\citenamefont {Zhou}\ and\ \citenamefont
  {Liu}(2024)}]{zhou2024hybrid}%
  \BibitemOpen
  \bibfield  {author} {\bibinfo {author} {\bibfnamefont {Y.}~\bibnamefont
  {Zhou}}\ and\ \bibinfo {author} {\bibfnamefont {Z.}~\bibnamefont {Liu}},\
  }\bibfield  {title} {\bibinfo {title} {A hybrid framework for estimating
  nonlinear functions of quantum states},\ }\href
  {https://www.nature.com/articles/s41534-024-00846-5} {\bibfield  {journal}
  {\bibinfo  {journal} {npj Quantum Inf.}\ }\textbf {\bibinfo {volume} {10}},\
  \bibinfo {pages} {62} (\bibinfo {year} {2024})}\BibitemShut {NoStop}%
\bibitem [{\citenamefont {Flammia}\ and\ \citenamefont
  {Liu}(2011)}]{Flammia2011Direct}%
  \BibitemOpen
  \bibfield  {author} {\bibinfo {author} {\bibfnamefont {S.~T.}\ \bibnamefont
  {Flammia}}\ and\ \bibinfo {author} {\bibfnamefont {Y.-K.}\ \bibnamefont
  {Liu}},\ }\bibfield  {title} {\bibinfo {title} {Direct fidelity estimation
  from few pauli measurements},\ }\href
  {https://doi.org/10.1103/PhysRevLett.106.230501} {\bibfield  {journal}
  {\bibinfo  {journal} {Phys. Rev. Lett.}\ }\textbf {\bibinfo {volume} {106}},\
  \bibinfo {pages} {230501} (\bibinfo {year} {2011})}\BibitemShut {NoStop}%
\bibitem [{\citenamefont {G\"uhne}\ \emph {et~al.}(2007)\citenamefont
  {G\"uhne}, \citenamefont {Lu}, \citenamefont {Gao},\ and\ \citenamefont
  {Pan}}]{Guhne2007toolbox}%
  \BibitemOpen
  \bibfield  {author} {\bibinfo {author} {\bibfnamefont {O.}~\bibnamefont
  {G\"uhne}}, \bibinfo {author} {\bibfnamefont {C.-Y.}\ \bibnamefont {Lu}},
  \bibinfo {author} {\bibfnamefont {W.-B.}\ \bibnamefont {Gao}},\ and\ \bibinfo
  {author} {\bibfnamefont {J.-W.}\ \bibnamefont {Pan}},\ }\bibfield  {title}
  {\bibinfo {title} {Toolbox for entanglement detection and fidelity
  estimation},\ }\href {https://doi.org/10.1103/PhysRevA.76.030305} {\bibfield
  {journal} {\bibinfo  {journal} {Phys. Rev. A}\ }\textbf {\bibinfo {volume}
  {76}},\ \bibinfo {pages} {030305} (\bibinfo {year} {2007})}\BibitemShut
  {NoStop}%
\bibitem [{\citenamefont {Seshadri}\ \emph {et~al.}(2024)\citenamefont
  {Seshadri}, \citenamefont {Ringbauer}, \citenamefont {Spainhour},
  \citenamefont {Blatt}, \citenamefont {Monz},\ and\ \citenamefont
  {Becker}}]{Seshadri2024versatile}%
  \BibitemOpen
  \bibfield  {author} {\bibinfo {author} {\bibfnamefont {A.}~\bibnamefont
  {Seshadri}}, \bibinfo {author} {\bibfnamefont {M.}~\bibnamefont {Ringbauer}},
  \bibinfo {author} {\bibfnamefont {J.}~\bibnamefont {Spainhour}}, \bibinfo
  {author} {\bibfnamefont {R.}~\bibnamefont {Blatt}}, \bibinfo {author}
  {\bibfnamefont {T.}~\bibnamefont {Monz}},\ and\ \bibinfo {author}
  {\bibfnamefont {S.}~\bibnamefont {Becker}},\ }\bibfield  {title} {\bibinfo
  {title} {Versatile fidelity estimation with confidence},\ }\href
  {https://doi.org/10.1103/PhysRevLett.133.020402} {\bibfield  {journal}
  {\bibinfo  {journal} {Phys. Rev. Lett.}\ }\textbf {\bibinfo {volume} {133}},\
  \bibinfo {pages} {020402} (\bibinfo {year} {2024})}\BibitemShut {NoStop}%
\bibitem [{\citenamefont {Pallister}\ \emph {et~al.}(2018)\citenamefont
  {Pallister}, \citenamefont {Linden},\ and\ \citenamefont
  {Montanaro}}]{Sam2018Optimal}%
  \BibitemOpen
  \bibfield  {author} {\bibinfo {author} {\bibfnamefont {S.}~\bibnamefont
  {Pallister}}, \bibinfo {author} {\bibfnamefont {N.}~\bibnamefont {Linden}},\
  and\ \bibinfo {author} {\bibfnamefont {A.}~\bibnamefont {Montanaro}},\
  }\bibfield  {title} {\bibinfo {title} {Optimal verification of entangled
  states with local measurements},\ }\href
  {https://doi.org/10.1103/PhysRevLett.120.170502} {\bibfield  {journal}
  {\bibinfo  {journal} {Phys. Rev. Lett.}\ }\textbf {\bibinfo {volume} {120}},\
  \bibinfo {pages} {170502} (\bibinfo {year} {2018})}\BibitemShut {NoStop}%
\bibitem [{\citenamefont {Hayashi}\ \emph {et~al.}(2006)\citenamefont
  {Hayashi}, \citenamefont {Matsumoto},\ and\ \citenamefont
  {Tsuda}}]{hayashi2006study}%
  \BibitemOpen
  \bibfield  {author} {\bibinfo {author} {\bibfnamefont {M.}~\bibnamefont
  {Hayashi}}, \bibinfo {author} {\bibfnamefont {K.}~\bibnamefont {Matsumoto}},\
  and\ \bibinfo {author} {\bibfnamefont {Y.}~\bibnamefont {Tsuda}},\ }\bibfield
   {title} {\bibinfo {title} {A study of locc-detection of a maximally
  entangled state using hypothesis testing},\ }\href
  {https://iopscience.iop.org/article/10.1088/0305-4470/39/46/013/meta}
  {\bibfield  {journal} {\bibinfo  {journal} {Journal of Physics A:
  Mathematical and General}\ }\textbf {\bibinfo {volume} {39}},\ \bibinfo
  {pages} {14427} (\bibinfo {year} {2006})}\BibitemShut {NoStop}%
\bibitem [{\citenamefont {Zhu}\ and\ \citenamefont
  {Hayashi}(2019{\natexlab{a}})}]{Zhu2019Efficient}%
  \BibitemOpen
  \bibfield  {author} {\bibinfo {author} {\bibfnamefont {H.}~\bibnamefont
  {Zhu}}\ and\ \bibinfo {author} {\bibfnamefont {M.}~\bibnamefont {Hayashi}},\
  }\bibfield  {title} {\bibinfo {title} {Efficient verification of pure quantum
  states in the adversarial scenario},\ }\href
  {https://doi.org/10.1103/PhysRevLett.123.260504} {\bibfield  {journal}
  {\bibinfo  {journal} {Phys. Rev. Lett.}\ }\textbf {\bibinfo {volume} {123}},\
  \bibinfo {pages} {260504} (\bibinfo {year} {2019}{\natexlab{a}})}\BibitemShut
  {NoStop}%
\bibitem [{\citenamefont {Zhu}\ and\ \citenamefont
  {Hayashi}(2019{\natexlab{b}})}]{Zhu2019General}%
  \BibitemOpen
  \bibfield  {author} {\bibinfo {author} {\bibfnamefont {H.}~\bibnamefont
  {Zhu}}\ and\ \bibinfo {author} {\bibfnamefont {M.}~\bibnamefont {Hayashi}},\
  }\bibfield  {title} {\bibinfo {title} {General framework for verifying pure
  quantum states in the adversarial scenario},\ }\href
  {https://doi.org/10.1103/PhysRevA.100.062335} {\bibfield  {journal} {\bibinfo
   {journal} {Phys. Rev. A}\ }\textbf {\bibinfo {volume} {100}},\ \bibinfo
  {pages} {062335} (\bibinfo {year} {2019}{\natexlab{b}})}\BibitemShut
  {NoStop}%
\bibitem [{\citenamefont {Yu}\ \emph {et~al.}(2019)\citenamefont {Yu},
  \citenamefont {Shang},\ and\ \citenamefont {G{\"u}hne}}]{yu2019optimal}%
  \BibitemOpen
  \bibfield  {author} {\bibinfo {author} {\bibfnamefont {X.-D.}\ \bibnamefont
  {Yu}}, \bibinfo {author} {\bibfnamefont {J.}~\bibnamefont {Shang}},\ and\
  \bibinfo {author} {\bibfnamefont {O.}~\bibnamefont {G{\"u}hne}},\ }\bibfield
  {title} {\bibinfo {title} {Optimal verification of general bipartite pure
  states},\ }\href {https://www.nature.com/articles/s41534-019-0226-z}
  {\bibfield  {journal} {\bibinfo  {journal} {npj Quantum Information}\
  }\textbf {\bibinfo {volume} {5}},\ \bibinfo {pages} {1} (\bibinfo {year}
  {2019})}\BibitemShut {NoStop}%
\bibitem [{\citenamefont {Liu}\ \emph {et~al.}(2019)\citenamefont {Liu},
  \citenamefont {Yu}, \citenamefont {Shang}, \citenamefont {Zhu},\ and\
  \citenamefont {Zhang}}]{Liu2019Dicke}%
  \BibitemOpen
  \bibfield  {author} {\bibinfo {author} {\bibfnamefont {Y.-C.}\ \bibnamefont
  {Liu}}, \bibinfo {author} {\bibfnamefont {X.-D.}\ \bibnamefont {Yu}},
  \bibinfo {author} {\bibfnamefont {J.}~\bibnamefont {Shang}}, \bibinfo
  {author} {\bibfnamefont {H.}~\bibnamefont {Zhu}},\ and\ \bibinfo {author}
  {\bibfnamefont {X.}~\bibnamefont {Zhang}},\ }\bibfield  {title} {\bibinfo
  {title} {Efficient verification of dicke states},\ }\href
  {https://doi.org/10.1103/PhysRevApplied.12.044020} {\bibfield  {journal}
  {\bibinfo  {journal} {Phys. Rev. Applied}\ }\textbf {\bibinfo {volume}
  {12}},\ \bibinfo {pages} {044020} (\bibinfo {year} {2019})}\BibitemShut
  {NoStop}%
\bibitem [{\citenamefont {Liu}\ \emph {et~al.}(2021{\natexlab{a}})\citenamefont
  {Liu}, \citenamefont {Shang}, \citenamefont {Han},\ and\ \citenamefont
  {Zhang}}]{Liu2021Nondemolition}%
  \BibitemOpen
  \bibfield  {author} {\bibinfo {author} {\bibfnamefont {Y.-C.}\ \bibnamefont
  {Liu}}, \bibinfo {author} {\bibfnamefont {J.}~\bibnamefont {Shang}}, \bibinfo
  {author} {\bibfnamefont {R.}~\bibnamefont {Han}},\ and\ \bibinfo {author}
  {\bibfnamefont {X.}~\bibnamefont {Zhang}},\ }\bibfield  {title} {\bibinfo
  {title} {Universally optimal verification of entangled states with
  nondemolition measurements},\ }\href
  {https://doi.org/10.1103/PhysRevLett.126.090504} {\bibfield  {journal}
  {\bibinfo  {journal} {Phys. Rev. Lett.}\ }\textbf {\bibinfo {volume} {126}},\
  \bibinfo {pages} {090504} (\bibinfo {year} {2021}{\natexlab{a}})}\BibitemShut
  {NoStop}%
\bibitem [{\citenamefont {Dangniam}\ \emph {et~al.}(2020)\citenamefont
  {Dangniam}, \citenamefont {Han},\ and\ \citenamefont
  {Zhu}}]{Dangniam2020Stabilizer}%
  \BibitemOpen
  \bibfield  {author} {\bibinfo {author} {\bibfnamefont {N.}~\bibnamefont
  {Dangniam}}, \bibinfo {author} {\bibfnamefont {Y.-G.}\ \bibnamefont {Han}},\
  and\ \bibinfo {author} {\bibfnamefont {H.}~\bibnamefont {Zhu}},\ }\bibfield
  {title} {\bibinfo {title} {Optimal verification of stabilizer states},\
  }\href {https://doi.org/10.1103/PhysRevResearch.2.043323} {\bibfield
  {journal} {\bibinfo  {journal} {Phys. Rev. Research}\ }\textbf {\bibinfo
  {volume} {2}},\ \bibinfo {pages} {043323} (\bibinfo {year}
  {2020})}\BibitemShut {NoStop}%
\bibitem [{\citenamefont {Huang}\ \emph {et~al.}(2024)\citenamefont {Huang},
  \citenamefont {Preskill},\ and\ \citenamefont
  {Soleimanifar}}]{Huang2024Certifying}%
  \BibitemOpen
  \bibfield  {author} {\bibinfo {author} {\bibfnamefont {H.-Y.}\ \bibnamefont
  {Huang}}, \bibinfo {author} {\bibfnamefont {J.}~\bibnamefont {Preskill}},\
  and\ \bibinfo {author} {\bibfnamefont {M.}~\bibnamefont {Soleimanifar}},\
  }\href@noop {} {\bibinfo {title} {Certifying almost all quantum states with
  few single-qubit measurements}} (\bibinfo {year} {2024}),\ \Eprint
  {https://arxiv.org/abs/2404.07281} {arXiv:2404.07281} \BibitemShut {NoStop}%
\bibitem [{\citenamefont {Zhu}\ and\ \citenamefont
  {Hayashi}(2019{\natexlab{c}})}]{Zhu2019Hypergraph}%
  \BibitemOpen
  \bibfield  {author} {\bibinfo {author} {\bibfnamefont {H.}~\bibnamefont
  {Zhu}}\ and\ \bibinfo {author} {\bibfnamefont {M.}~\bibnamefont {Hayashi}},\
  }\bibfield  {title} {\bibinfo {title} {Efficient verification of hypergraph
  states},\ }\href {https://doi.org/10.1103/PhysRevApplied.12.054047}
  {\bibfield  {journal} {\bibinfo  {journal} {Phys. Rev. Applied}\ }\textbf
  {\bibinfo {volume} {12}},\ \bibinfo {pages} {054047} (\bibinfo {year}
  {2019}{\natexlab{c}})}\BibitemShut {NoStop}%
\bibitem [{\citenamefont {Calderbank}\ and\ \citenamefont
  {Shor}(1996)}]{Calderbank1996CSScode}%
  \BibitemOpen
  \bibfield  {author} {\bibinfo {author} {\bibfnamefont {A.~R.}\ \bibnamefont
  {Calderbank}}\ and\ \bibinfo {author} {\bibfnamefont {P.~W.}\ \bibnamefont
  {Shor}},\ }\bibfield  {title} {\bibinfo {title} {Good quantum
  error-correcting codes exist},\ }\href
  {https://doi.org/10.1103/PhysRevA.54.1098} {\bibfield  {journal} {\bibinfo
  {journal} {Phys. Rev. A}\ }\textbf {\bibinfo {volume} {54}},\ \bibinfo
  {pages} {1098} (\bibinfo {year} {1996})}\BibitemShut {NoStop}%
\bibitem [{\citenamefont {Steane}(1996)}]{Steane1996CSScode}%
  \BibitemOpen
  \bibfield  {author} {\bibinfo {author} {\bibfnamefont {A.}~\bibnamefont
  {Steane}},\ }\bibfield  {title} {\bibinfo {title} {Multiple-particle
  interference and quantum error correction},\ }\href
  {https://royalsocietypublishing.org/doi/10.1098/rspa.1996.0136} {\bibfield
  {journal} {\bibinfo  {journal} {Proceedings of the Royal Society of London.
  Series A: Mathematical, Physical and Engineering Sciences}\ }\textbf
  {\bibinfo {volume} {452}},\ \bibinfo {pages} {2551} (\bibinfo {year}
  {1996})}\BibitemShut {NoStop}%
\bibitem [{\citenamefont {Anshu}\ and\ \citenamefont
  {Arunachalam}(2024)}]{Anshu2024Survey}%
  \BibitemOpen
  \bibfield  {author} {\bibinfo {author} {\bibfnamefont {A.}~\bibnamefont
  {Anshu}}\ and\ \bibinfo {author} {\bibfnamefont {S.}~\bibnamefont
  {Arunachalam}},\ }\bibfield  {title} {\bibinfo {title} {A survey on the
  complexity of learning quantum states},\ }\href
  {https://www.nature.com/articles/s42254-023-00662-4} {\bibfield  {journal}
  {\bibinfo  {journal} {Nature Reviews Physics}\ }\textbf {\bibinfo {volume}
  {6}},\ \bibinfo {pages} {59} (\bibinfo {year} {2024})}\BibitemShut {NoStop}%
\bibitem [{\citenamefont {T\'oth}\ and\ \citenamefont
  {G\"uhne}(2005{\natexlab{a}})}]{Toth2005Detecting}%
  \BibitemOpen
  \bibfield  {author} {\bibinfo {author} {\bibfnamefont {G.}~\bibnamefont
  {T\'oth}}\ and\ \bibinfo {author} {\bibfnamefont {O.}~\bibnamefont
  {G\"uhne}},\ }\bibfield  {title} {\bibinfo {title} {Detecting genuine
  multipartite entanglement with two local measurements},\ }\href
  {https://doi.org/10.1103/PhysRevLett.94.060501} {\bibfield  {journal}
  {\bibinfo  {journal} {Phys. Rev. Lett.}\ }\textbf {\bibinfo {volume} {94}},\
  \bibinfo {pages} {060501} (\bibinfo {year} {2005}{\natexlab{a}})}\BibitemShut
  {NoStop}%
\bibitem [{\citenamefont {T\'oth}\ and\ \citenamefont
  {G\"uhne}(2005{\natexlab{b}})}]{Toth2005stabilizer}%
  \BibitemOpen
  \bibfield  {author} {\bibinfo {author} {\bibfnamefont {G.}~\bibnamefont
  {T\'oth}}\ and\ \bibinfo {author} {\bibfnamefont {O.}~\bibnamefont
  {G\"uhne}},\ }\bibfield  {title} {\bibinfo {title} {Entanglement detection in
  the stabilizer formalism},\ }\href
  {https://doi.org/10.1103/PhysRevA.72.022340} {\bibfield  {journal} {\bibinfo
  {journal} {Phys. Rev. A}\ }\textbf {\bibinfo {volume} {72}},\ \bibinfo
  {pages} {022340} (\bibinfo {year} {2005}{\natexlab{b}})}\BibitemShut
  {NoStop}%
\bibitem [{\citenamefont {Zhou}\ \emph {et~al.}(2019)\citenamefont {Zhou},
  \citenamefont {Zhao}, \citenamefont {Yuan},\ and\ \citenamefont
  {Ma}}]{Zhou2019structure}%
  \BibitemOpen
  \bibfield  {author} {\bibinfo {author} {\bibfnamefont {Y.}~\bibnamefont
  {Zhou}}, \bibinfo {author} {\bibfnamefont {Q.}~\bibnamefont {Zhao}}, \bibinfo
  {author} {\bibfnamefont {X.}~\bibnamefont {Yuan}},\ and\ \bibinfo {author}
  {\bibfnamefont {X.}~\bibnamefont {Ma}},\ }\bibfield  {title} {\bibinfo
  {title} {Detecting multipartite entanglement structure with minimal
  resources},\ }\href {https://doi.org/10.1038/s41534-019-0200-9} {\bibfield
  {journal} {\bibinfo  {journal} {npj Quantum Information}\ }\textbf {\bibinfo
  {volume} {5}},\ \bibinfo {pages} {83} (\bibinfo {year} {2019})}\BibitemShut
  {NoStop}%
\bibitem [{\citenamefont {Zhao}\ and\ \citenamefont
  {Zhou}(2022)}]{zhao2020constructing}%
  \BibitemOpen
  \bibfield  {author} {\bibinfo {author} {\bibfnamefont {Q.}~\bibnamefont
  {Zhao}}\ and\ \bibinfo {author} {\bibfnamefont {Y.}~\bibnamefont {Zhou}},\
  }\bibfield  {title} {\bibinfo {title} {Constructing multipartite bell
  inequalities from stabilizers},\ }\href
  {https://doi.org/10.1103/PhysRevResearch.4.043215} {\bibfield  {journal}
  {\bibinfo  {journal} {Phys. Rev. Res.}\ }\textbf {\bibinfo {volume} {4}},\
  \bibinfo {pages} {043215} (\bibinfo {year} {2022})}\BibitemShut {NoStop}%
\bibitem [{\citenamefont {Hein}\ \emph {et~al.}(2006)\citenamefont {Hein},
  \citenamefont {Dür}, \citenamefont {Eisert}, \citenamefont {Raussendorf},
  \citenamefont {den Nest},\ and\ \citenamefont {Briegel}}]{Hein2006Graph}%
  \BibitemOpen
  \bibfield  {author} {\bibinfo {author} {\bibfnamefont {M.}~\bibnamefont
  {Hein}}, \bibinfo {author} {\bibfnamefont {W.}~\bibnamefont {Dür}}, \bibinfo
  {author} {\bibfnamefont {J.}~\bibnamefont {Eisert}}, \bibinfo {author}
  {\bibfnamefont {R.}~\bibnamefont {Raussendorf}}, \bibinfo {author}
  {\bibfnamefont {M.~V.}\ \bibnamefont {den Nest}},\ and\ \bibinfo {author}
  {\bibfnamefont {H.~J.}\ \bibnamefont {Briegel}},\ }\href@noop {} {\bibinfo
  {title} {Entanglement in graph states and its applications}} (\bibinfo {year}
  {2006}),\ \Eprint {https://arxiv.org/abs/quant-ph/0602096}
  {arXiv:quant-ph/0602096 [quant-ph]} \BibitemShut {NoStop}%
\bibitem [{\citenamefont {Kitaev}(2003{\natexlab{b}})}]{KITAEV2003Fault}%
  \BibitemOpen
  \bibfield  {author} {\bibinfo {author} {\bibfnamefont {A.}~\bibnamefont
  {Kitaev}},\ }\bibfield  {title} {\bibinfo {title} {Fault-tolerant quantum
  computation by anyons},\ }\href
  {https://doi.org/https://doi.org/10.1016/S0003-4916(02)00018-0} {\bibfield
  {journal} {\bibinfo  {journal} {Annals of Physics}\ }\textbf {\bibinfo
  {volume} {303}},\ \bibinfo {pages} {2 } (\bibinfo {year}
  {2003}{\natexlab{b}})}\BibitemShut {NoStop}%
\bibitem [{\citenamefont {Tillich}\ and\ \citenamefont
  {Zémor}(2014)}]{Tillich2014HGPcode}%
  \BibitemOpen
  \bibfield  {author} {\bibinfo {author} {\bibfnamefont {J.-P.}\ \bibnamefont
  {Tillich}}\ and\ \bibinfo {author} {\bibfnamefont {G.}~\bibnamefont
  {Zémor}},\ }\bibfield  {title} {\bibinfo {title} {Quantum ldpc codes with
  positive rate and minimum distance proportional to the square root of the
  blocklength},\ }\href {https://doi.org/10.1109/TIT.2013.2292061} {\bibfield
  {journal} {\bibinfo  {journal} {IEEE Transactions on Information Theory}\
  }\textbf {\bibinfo {volume} {60}},\ \bibinfo {pages} {1193} (\bibinfo {year}
  {2014})}\BibitemShut {NoStop}%
\bibitem [{\citenamefont {Gottesman}\ \emph {et~al.}(2001)\citenamefont
  {Gottesman}, \citenamefont {Kitaev},\ and\ \citenamefont
  {Preskill}}]{Gottesman2001GKPcode}%
  \BibitemOpen
  \bibfield  {author} {\bibinfo {author} {\bibfnamefont {D.}~\bibnamefont
  {Gottesman}}, \bibinfo {author} {\bibfnamefont {A.}~\bibnamefont {Kitaev}},\
  and\ \bibinfo {author} {\bibfnamefont {J.}~\bibnamefont {Preskill}},\
  }\bibfield  {title} {\bibinfo {title} {Encoding a qubit in an oscillator},\
  }\href {https://doi.org/10.1103/PhysRevA.64.012310} {\bibfield  {journal}
  {\bibinfo  {journal} {Phys. Rev. A}\ }\textbf {\bibinfo {volume} {64}},\
  \bibinfo {pages} {012310} (\bibinfo {year} {2001})}\BibitemShut {NoStop}%
\bibitem [{\citenamefont {Liu}\ \emph {et~al.}(2021{\natexlab{b}})\citenamefont
  {Liu}, \citenamefont {Shang},\ and\ \citenamefont {Zhang}}]{Liu2021CV}%
  \BibitemOpen
  \bibfield  {author} {\bibinfo {author} {\bibfnamefont {Y.-C.}\ \bibnamefont
  {Liu}}, \bibinfo {author} {\bibfnamefont {J.}~\bibnamefont {Shang}},\ and\
  \bibinfo {author} {\bibfnamefont {X.}~\bibnamefont {Zhang}},\ }\bibfield
  {title} {\bibinfo {title} {Efficient verification of entangled
  continuous-variable quantum states with local measurements},\ }\href
  {https://doi.org/10.1103/PhysRevResearch.3.L042004} {\bibfield  {journal}
  {\bibinfo  {journal} {Phys. Rev. Research}\ }\textbf {\bibinfo {volume}
  {3}},\ \bibinfo {pages} {L042004} (\bibinfo {year}
  {2021}{\natexlab{b}})}\BibitemShut {NoStop}%
\bibitem [{\citenamefont {Leung}\ \emph {et~al.}(1997)\citenamefont {Leung},
  \citenamefont {Nielsen}, \citenamefont {Chuang},\ and\ \citenamefont
  {Yamamoto}}]{Leung1997Approximate}%
  \BibitemOpen
  \bibfield  {author} {\bibinfo {author} {\bibfnamefont {D.~W.}\ \bibnamefont
  {Leung}}, \bibinfo {author} {\bibfnamefont {M.~A.}\ \bibnamefont {Nielsen}},
  \bibinfo {author} {\bibfnamefont {I.~L.}\ \bibnamefont {Chuang}},\ and\
  \bibinfo {author} {\bibfnamefont {Y.}~\bibnamefont {Yamamoto}},\ }\bibfield
  {title} {\bibinfo {title} {Approximate quantum error correction can lead to
  better codes},\ }\href {https://doi.org/10.1103/PhysRevA.56.2567} {\bibfield
  {journal} {\bibinfo  {journal} {Phys. Rev. A}\ }\textbf {\bibinfo {volume}
  {56}},\ \bibinfo {pages} {2567} (\bibinfo {year} {1997})}\BibitemShut
  {NoStop}%
\bibitem [{\citenamefont {Hayden}\ and\ \citenamefont
  {Penington}(2020)}]{Hayden2020Approximate}%
  \BibitemOpen
  \bibfield  {author} {\bibinfo {author} {\bibfnamefont {P.}~\bibnamefont
  {Hayden}}\ and\ \bibinfo {author} {\bibfnamefont {G.}~\bibnamefont
  {Penington}},\ }\bibfield  {title} {\bibinfo {title} {Approximate quantum
  error correction revisited: introducing the alpha-bit},\ }\href
  {https://link.springer.com/article/10.1007/s00220-020-03689-1} {\bibfield
  {journal} {\bibinfo  {journal} {Communications in Mathematical Physics}\
  }\textbf {\bibinfo {volume} {374}},\ \bibinfo {pages} {369} (\bibinfo {year}
  {2020})}\BibitemShut {NoStop}%
\bibitem [{\citenamefont {Liu}\ \emph {et~al.}(2020)\citenamefont {Liu},
  \citenamefont {Shang}, \citenamefont {Yu},\ and\ \citenamefont
  {Zhang}}]{liu2019efficient}%
  \BibitemOpen
  \bibfield  {author} {\bibinfo {author} {\bibfnamefont {Y.-C.}\ \bibnamefont
  {Liu}}, \bibinfo {author} {\bibfnamefont {J.}~\bibnamefont {Shang}}, \bibinfo
  {author} {\bibfnamefont {X.-D.}\ \bibnamefont {Yu}},\ and\ \bibinfo {author}
  {\bibfnamefont {X.}~\bibnamefont {Zhang}},\ }\bibfield  {title} {\bibinfo
  {title} {Efficient verification of quantum processes},\ }\href
  {https://doi.org/10.1103/PhysRevA.101.042315} {\bibfield  {journal} {\bibinfo
   {journal} {Phys. Rev. A}\ }\textbf {\bibinfo {volume} {101}},\ \bibinfo
  {pages} {042315} (\bibinfo {year} {2020})}\BibitemShut {NoStop}%
\bibitem [{\citenamefont {Zhu}\ and\ \citenamefont
  {Zhang}(2020)}]{zhu2019efficientgate}%
  \BibitemOpen
  \bibfield  {author} {\bibinfo {author} {\bibfnamefont {H.}~\bibnamefont
  {Zhu}}\ and\ \bibinfo {author} {\bibfnamefont {H.}~\bibnamefont {Zhang}},\
  }\bibfield  {title} {\bibinfo {title} {Efficient verification of quantum
  gates with local operations},\ }\href
  {https://doi.org/10.1103/PhysRevA.101.042316} {\bibfield  {journal} {\bibinfo
   {journal} {Phys. Rev. A}\ }\textbf {\bibinfo {volume} {101}},\ \bibinfo
  {pages} {042316} (\bibinfo {year} {2020})}\BibitemShut {NoStop}%
\bibitem [{\citenamefont {Zeng}\ \emph {et~al.}(2020)\citenamefont {Zeng},
  \citenamefont {Zhou},\ and\ \citenamefont {Liu}}]{Zeng2020gate}%
  \BibitemOpen
  \bibfield  {author} {\bibinfo {author} {\bibfnamefont {P.}~\bibnamefont
  {Zeng}}, \bibinfo {author} {\bibfnamefont {Y.}~\bibnamefont {Zhou}},\ and\
  \bibinfo {author} {\bibfnamefont {Z.}~\bibnamefont {Liu}},\ }\bibfield
  {title} {\bibinfo {title} {Quantum gate verification and its application in
  property testing},\ }\href {https://doi.org/10.1103/PhysRevResearch.2.023306}
  {\bibfield  {journal} {\bibinfo  {journal} {Phys. Rev. Research}\ }\textbf
  {\bibinfo {volume} {2}},\ \bibinfo {pages} {023306} (\bibinfo {year}
  {2020})}\BibitemShut {NoStop}%
\bibitem [{\citenamefont {Gebhart}\ \emph {et~al.}(2023)\citenamefont
  {Gebhart}, \citenamefont {Santagati}, \citenamefont {Gentile}, \citenamefont
  {Gauger}, \citenamefont {Craig}, \citenamefont {Ares}, \citenamefont
  {Banchi}, \citenamefont {Marquardt}, \citenamefont {Pezz{\`e}},\ and\
  \citenamefont {Bonato}}]{gebhart2023learning}%
  \BibitemOpen
  \bibfield  {author} {\bibinfo {author} {\bibfnamefont {V.}~\bibnamefont
  {Gebhart}}, \bibinfo {author} {\bibfnamefont {R.}~\bibnamefont {Santagati}},
  \bibinfo {author} {\bibfnamefont {A.~A.}\ \bibnamefont {Gentile}}, \bibinfo
  {author} {\bibfnamefont {E.~M.}\ \bibnamefont {Gauger}}, \bibinfo {author}
  {\bibfnamefont {D.}~\bibnamefont {Craig}}, \bibinfo {author} {\bibfnamefont
  {N.}~\bibnamefont {Ares}}, \bibinfo {author} {\bibfnamefont {L.}~\bibnamefont
  {Banchi}}, \bibinfo {author} {\bibfnamefont {F.}~\bibnamefont {Marquardt}},
  \bibinfo {author} {\bibfnamefont {L.}~\bibnamefont {Pezz{\`e}}},\ and\
  \bibinfo {author} {\bibfnamefont {C.}~\bibnamefont {Bonato}},\ }\bibfield
  {title} {\bibinfo {title} {Learning quantum systems},\ }\href
  {https://www.nature.com/articles/s42254-022-00552-1} {\bibfield  {journal}
  {\bibinfo  {journal} {Nature Reviews Physics}\ }\textbf {\bibinfo {volume}
  {5}},\ \bibinfo {pages} {141} (\bibinfo {year} {2023})}\BibitemShut {NoStop}%
\bibitem [{\citenamefont {Zheng}\ \emph {et~al.}(2024)\citenamefont {Zheng},
  \citenamefont {Yu}, \citenamefont {Zhang}, \citenamefont {Xu},\ and\
  \citenamefont {Wang}}]{zheng2024efficient}%
  \BibitemOpen
  \bibfield  {author} {\bibinfo {author} {\bibfnamefont {C.}~\bibnamefont
  {Zheng}}, \bibinfo {author} {\bibfnamefont {X.}~\bibnamefont {Yu}}, \bibinfo
  {author} {\bibfnamefont {Z.}~\bibnamefont {Zhang}}, \bibinfo {author}
  {\bibfnamefont {P.}~\bibnamefont {Xu}},\ and\ \bibinfo {author}
  {\bibfnamefont {K.}~\bibnamefont {Wang}},\ }\href@noop {} {\bibinfo {title}
  {Efficient verification of stabilizer code subspaces with local
  measurements}} (\bibinfo {year} {2024}),\ \Eprint
  {https://arxiv.org/abs/2409.19699} {arXiv:2409.19699 [quant-ph]} \BibitemShut
  {NoStop}%
\end{thebibliography}%

\appendix
\onecolumngrid
\newpage

\section{Proof of Theorem \ref{theo:verification}}\label{app:prooftheo1}
First, the following lemma gives the range of the passing probability $\Tr{\Omega\rho}$ of a single test.

\begin{lemma}\label{lemma:passprob}
    The passing probability of $\rho$ of a single test satisfies 
    \begin{equation}
    \begin{aligned}
      1-\Delta_{\max}(\Omega)\epsilon_{\rho} \leq \Tr{\Omega\rho}\leq 1-\Delta_{\min}(\Omega)\epsilon_{\rho},
    \end{aligned}
    \end{equation}
where $\epsilon_{\rho}:=1-F(\mathcal{P},\rho)$ is the infidelity of $\rho$ to subspace $\mathcal{V}$.
\end{lemma}

\begin{proof}
For simplicity, we only prove the upper bound. The lower bound can be proved in the same way. By definition and the decomposition in Eq.~\eqref{eq:decom},
\begin{equation}
\begin{aligned}
\Tr{\Omega\rho}
&=\Tr{\mathcal{P}\rho}+\Tr{\sum_{j=d_{\mathcal{V}}+1}^{d}\lambda_j \ketbra{\phi_j}{\phi_j}\rho}\\
&\leq 1-\epsilon_{\rho}+\lambda_{d_{\mathcal{V}}+1}\sum_{j=d_{\mathcal{V}}+1}^{d}\Tr{\ketbra{\phi_j}{\phi_j}\rho} \\
&=1-\epsilon_{\rho}+\lambda_{d_{\mathcal{V}}+1}\epsilon_{\rho}=1-\Delta_{\min}(\Omega)\epsilon_{\rho}.
\end{aligned}
\end{equation}
Here we use the fact that $\sum_{j=d_{\mathcal{V}}+1}^{d}\ketbra{\phi_j}{\phi_j}=\id-\mathcal{P}$ is the projector on $\mathcal{V}^{\bot}$. Note that the inequality can be saturated as all the nonzero $\Tr{\ketbra{\phi_j}{\phi_j}\rho}$ corresponds to eigenvalues equaling to $\lambda_{d_{\mathcal{V}}+1}$.
\end{proof}

Based on the former lemma, the expected total number of passed tests satisfies 
\begin{equation}
\begin{aligned}
  \mathbb{E}[N_{\text{pass}}]\in[\left(1-\Delta_{\max}(\Omega)\epsilon_{\rho}\right)N,\left(1-\Delta_{\min}(\Omega)\epsilon_{\rho}\right)N].
\end{aligned}
\end{equation}

If the device is ``good", there should be $\epsilon_{\rho}\leq \tau\epsilon$. By adopting the Chernoff–Hoeffding theorem, we can derive
\begin{equation}\label{eq:goodlow1}
\begin{aligned}
\mathrm{Pr}\left(N_{\text{pass}} \leq p_0N\right)&=\mathrm{Pr}\left(\frac{N_{\text{pass}}}{N}\leq \frac{\mathbb{E}[N_{\text{pass}}]}{N}-\frac{\mathbb{E}[N_{\text{pass}}]-p_0N}{N}\right)\\
&\leq e^{-D\left[p_0\|\frac{\mathbb{E}[N_{\text{pass}}]}{N}\right]N}\\
&\leq e^{-D\left[p_0\|1-\tilde{\epsilon}/r\right]N}=\delta.
\end{aligned}
\end{equation}
The third line is based on
\begin{equation}\label{eq:geq}
\frac{\mathbb{E}[N_{\text{pass}}]}{N}\geq 1-\Delta_{\max}(\Omega)\epsilon_{\rho}\geq 1-\tilde{\epsilon}/r\geq p_0.
\end{equation}
Since $p_0$ is the solution to the equation $D\left[p_0\|1-\tilde{\epsilon}/r\right]=D\left[p_0\|1-\tilde{\epsilon}\right]$, it should be between $1-\tilde{\epsilon}$ and $1-\tilde{\epsilon}/r$. The final inequality in Eq.~\eqref{eq:geq} comes from $1-\tilde{\epsilon}\leq p_0\leq 1-\tilde{\epsilon}/r$ when $r>1$.

If the device is ``bad", there should be $\epsilon_{\rho}\geq \epsilon$. We can similarly derive
\begin{equation}\label{eq:badhigh1}
\begin{aligned}
\mathrm{Pr}\left(N_{\text{pass}} \geq p_0N\right)&=\mathrm{Pr}\left(\frac{N_{\text{pass}}}{N}\geq \frac{\mathbb{E}[N_{\text{pass}}]}{N}+\frac{p_0N-\mathbb{E}[N_{\text{pass}}]}{N}\right)\\
&\leq e^{-D\left[p_0\|\frac{\mathbb{E}[N_{\text{pass}}]}{N}\right]N}\\
&\leq e^{-D\left[p_0\|1-\tilde{\epsilon}\right]N}=\delta.
\end{aligned}
\end{equation}
The third line is based on
\begin{equation}\label{eq:leq}
\frac{\mathbb{E}[N_{\text{pass}}]}{N}\leq 1-\Delta_{\min}(\Omega)\epsilon_{\rho}\leq 1-\tilde{\epsilon}\leq p_0.
\end{equation}

\section{Proof of Proposition \ref{prop:estimation}}\label{app:proofprop1}
First, we can give an interval estimator of the passing probability $\frac{\mathbb{E}[N_{\text{pass}}]}{N}$ based on the simplified confidence interval of the Bernoulli model:
\begin{equation}
\begin{aligned}
p-\xi\leq \frac{\mathbb{E}[N_{\text{pass}}]}{N}\leq p+\xi
\end{aligned}
\end{equation}
with approximate confidence level $1-\delta$.

From Lemma.~\ref{lemma:passprob}, $\frac{\mathbb{E}[N_{\text{pass}}]}{N}=\Tr{\Omega\rho}\in[1-\Delta_{\max}(\Omega)\epsilon_{\rho},1-\Delta_{\min}(\Omega)\epsilon_{\rho}]$. By combining these two conditions, with approximate confidence level $1-\delta$ we need $1-\Delta_{\max}(\Omega)\epsilon_{\rho}\leq p+\xi$ and $1-\Delta_{\min}(\Omega)\epsilon_{\rho}\geq p-\xi$, from which we can solve that 
\begin{equation}
\begin{aligned}
\frac{1-p-\xi}{\Delta_{\max}(\Omega)}\leq \epsilon_{\rho}\leq \frac{1-p+\xi}{\Delta_{\min}(\Omega)}.
\end{aligned}
\end{equation}
By combining it with $0\leq\epsilon_{\rho}\leq1$, we can derive the proposition.

\section{Spectral gaps of verification operators of stabilizer codes}\label{app:stabilizerdiscussion}
We can write $\Omega_{\text{all}}=\frac{\id+\mathcal{P}}{2}=\mathcal{P}+\frac{1}{2}\mathcal{P}^{\bot}$, where $\mathcal{P}^{\bot}$ represents the projector onto the subspace orthogonal to $\mathcal{V}$. This expression is just a diagonalization of $\Omega_{\text{all}}$, showing that the second largest and smallest eigenvalues are both $\frac{1}{2}$. Thus, The spectral gaps are $\Delta_{\min}(\Omega_{\text{all}})=\Delta_{\max}(\Omega_{\text{all}})=\frac{1}{2}$.

For $\Omega_{\text{gen}}$, consider the subspace $\mathcal{V}_{\vec{r}}$ with projector 
\begin{equation}
\begin{aligned}
\mathcal{P}_{\vec{r}}=\prod_{i=1}^m\frac{\id+(-1)^{\vec{r}(i)}S_i}{2}
\end{aligned}
\end{equation}
characterized by a binary vector $\vec{r}=(\cdots,\vec{r}(i),\cdots)$ of dimension $m$. It is clear that the code subspace $\mathcal{V}$ corresponds to $\vec{r}=\vec{0}$. One can verify their orthogonality by
\begin{equation}
\mathcal{P}_{\vec{r}}\mathcal{P}_{\vec{r}'}=\delta(\vec{r},\vec{r}')\mathcal{P}_{\vec{r}}
\end{equation}
where $\delta(a,b)=1$ if $a=b$ and $\delta(a,b)=0$ otherwise. Since $S_i\mathcal{P}_{\vec{r}}=(-1)^{\vec{r}(i)} \mathcal{P}_{\vec{r}}$, we can derive
\begin{equation}
\Omega_{\text{gen}}\mathcal{P}_{\vec{r}}=\frac{1}{m}\sum_{i=1}^m \frac{\id+S_i}{2}\mathcal{P}_{\vec{r}}=\left(1-\frac{1}{m}\sum_{i=1}^m \vec{r}(i)\right) \mathcal{P}_{\vec{r}}.
\end{equation}
Here, the calculation of $\sum_{i=1}^m \vec{r}(i)$ is performed on the field of $\mathbb{R}$ rather than $\mathbb{Z}_2$. 

For different $\vec{r}$, the eigenstates in subspace $\mathcal{P}_{\vec{r}}$ are always eigenstates of $\Omega_{\text{chr}}(\mathcal{S})$ with eigenvalues $1-\frac{1}{m}\sum_{i=1}^m \vec{r}(i)$. There are $2^m$ different $\vec{r}$, corresponding to $2^m$ orthogonal subspace $\mathcal{V}_{\vec{r}}$ of the whole Hilbert space. Each of them has dimension $2^k$. Therefore, $1-\frac{1}{m}\sum_{i=1}^m \vec{r}(i)$ are all the possible eigenvalues. The largest eigenvalue is $1$, corresponding to $\vec{r}=0$ with degeneracy $2^k$. The second largest and smallest eigenvalues are $1-\frac{1}{m}$ and $0$. The spectral gaps are $\Delta_{\min}(\Omega_{\text{gen}})=\frac{1}{m}$ and $\Delta_{\max}(\Omega_{\text{gen}})=1$.

For $\Omega_{\text{chr}}(\mathcal{S})$ in Eq.~\eqref{Eq:OmegaColor}, we still consider the subspace $\mathcal{V}_{\vec{r}}$ with projector $\mathcal{P}_{\vec{r}}$
\begin{equation}
\begin{aligned}
\mathcal{P}_{\vec{r}}=\prod_{i=1}^m\frac{\id+(-1)^{\vec{r}(i)}S_i}{2}=\prod_{u=1}^{\chi(G_{\mathcal{S}})}\prod_{v_i\in \mathcal{I}_u}\frac{\id+(-1)^{\vec{r}(i)}S_i}{2}=\prod_{u=1}^{\chi(G_{\mathcal{S}})}\mathcal{P}_{\mathcal{I}_u,\vec{r}_u}
\end{aligned}
\end{equation}
where $\mathcal{P}_{\mathcal{I}_u,\vec{r}_u}$ is defined as
\begin{equation}
\begin{aligned}
\mathcal{P}_{\mathcal{I}_u,\vec{r}_u}=\prod_{v_i\in \mathcal{I}_u}\frac{\id+(-1)^{r_u(i)}S_i}{2}
\end{aligned}
\end{equation}
depending on the restricted vector on the index of $\mathcal{I}_u$, i.e., $\vec{r}_u=\sum_{v_i\in \mathcal{I}_u}r(i)\vec{e}_i$ with $\sum_{u=1}^{\chi(G_{\mathcal{S}})}\vec{r}_u=\vec{r}$. Then, one has 
\begin{equation}
\begin{aligned}
\mathcal{P}_{\mathcal{I}_u,\vec{r}_u}\cdot \mathcal{P}_{\mathcal{I}_u}=\delta(\vec{r}_u,0)\mathcal{P}_{\mathcal{I}_u}.
\end{aligned}
\end{equation}
where $\mathcal{P}_{\mathcal{I}_u}$ is defined in Eq.~\eqref{eq:PIu} for the all-zero vector. As a result, 
\begin{equation}
\begin{aligned}
\Omega_{\text{chr}}(\mathcal{S}) \mathcal{P}_{\vec{r}}&=\frac{1}{\chi(G_{\mathcal{S}})}\sum_{u=1}^{\chi(G_{\mathcal{S}})}\mathcal{P}_{\mathcal{I}_u}\mathcal{P}_{\vec{r}}\\
&=\frac{1}{\chi(G_{\mathcal{S}})}\sum_{u=1}^{\chi(G_{\mathcal{S}})}\delta(\vec{r}_u,0)\mathcal{P}_{\vec{r}}.
\end{aligned}
\end{equation}
For different $\vec{r}$, the eigenstates in subspace $\mathcal{P}_{\vec{r}}$ are always eigenstates of $\Omega_{\text{chr}}(\mathcal{S})$ and their eigenvalue are 
\begin{equation}
\begin{aligned}
\lambda_{\vec{r}}=\frac{\sum_{u=1}^{\chi(G_{\mathcal{S}})}\delta(\vec{r}_u,0)}{\chi(G_{\mathcal{S}})}.
\end{aligned}
\end{equation}
There are $2^m$ different $\vec{r}$, corresponding to $2^m$ orthogonal subspace $\mathcal{V}_{\vec{r}}$ of the whole Hilbert space. Each of them has dimension $2^k$. Therefore, $\lambda_{\vec{r}}$ are all the eigenvalues of $\Omega_{\text{chr}}(\mathcal{S})$. Only when all the $\vec{r}_u=0$, i.e., $\vec{r}=0$ and $\mathcal{V}_{\vec{r}}$ is the code subspace, we have $\lambda_{\vec{r}}=1$. The second largest and smallest eigenvalues are $1-\frac{1}{\chi(G_{\mathcal{S}})}$ and $0$. The spectral gaps are $\Delta_{\min}(\Omega_{\text{chr}}(\mathcal{S}))=\frac{1}{\chi(G_{\mathcal{S}})}$ and $\Delta_{\max}(\Omega_{\text{chr}}(\mathcal{S}))=1$.

\section{Spectral gaps of verification operators of QLDPC codes}\label{app:qldpcdiscussion}
By replacing $S_i$ with $T_i$ in the analysis of $\Omega_{\text{gen}}$ and $\Omega_{\text{chr}}(\mathcal{S})$ in Appendix \ref{app:stabilizerdiscussion}, one can similarly derive the spectral gaps of $\Omega^{(s)}_{\text{gen}}$ and $\Omega^{(s)}_{\text{chr}}$.

For $\Omega^{(1)}_{\text{gen}}$ given in Eq.~\eqref{eq:LDPCgen1}, we can rewrite it as 
\begin{equation}
\Omega^{(1)}_{\text{gen}}=\frac{1}{m}\sum_i\left(\frac{1}{a_i}\pi_i+\left(1-\frac{1}{a_i}\right)\frac{\id}{2}\right)=\frac{1}{2m}\sum_i\left(\id+\frac{1}{a_i}T_i\right).
\end{equation}
With a similar construction of $\mathcal{P}_{\vec{r}}$ as that in Appendix \ref{app:stabilizerdiscussion}, we can derive its eigenvalues
\begin{equation}
\lambda_{\vec{r}}=\frac{1}{2m}\sum_{i=1}^m\left(1+\frac{(-1)^{\vec{r}(i)}}{a_i}\right)=\frac{1}{2}+\frac{1}{2m}\sum_{i=1}^m \frac{(-1)^{\vec{r}(i)}}{a_i}.
\end{equation}
Since $a_i>0$ for all $i$, the largest eigenvalue should be $\frac{1}{2}+\frac{1}{2m}\sum_{i=1}^m \frac{1}{a_i}$, the second largest eigenvalue should be $\frac{1}{2}+\frac{1}{2m}\left(\sum_{i=1}^m \frac{1}{a_i}-\frac{2}{\max_i a_i}\right)$, and the smallest eigenvalue should be $\frac{1}{2}-\frac{1}{2m}\sum_{i=1}^m \frac{1}{a_i}$. Thus, the spectral gaps are $\Delta_{\min}(\Omega^{(1)}_{\text{gen}})=\frac{1}{m\cdot \max_i a_i}$ and $\Delta_{\max}(\Omega^{(1)}_{\text{gen}})=\frac{1}{m}\sum_{i=1}^m \frac{1}{a_i}$.

Furthermore, we can construct
\begin{equation}
\begin{aligned}
\Omega_{\text{chr}}^{(1)}&=\frac{1}{\chi(G_{\Pi})}\sum_{u=1}^{\chi(G_{\Pi})}\prod_{v_i\in\mathcal{I}_u}\left(\frac{1}{a_i}\pi_i+\left(1-\frac{1}{a_i}\right)\frac{\id}{2}\right)\\
&=\frac{1}{\chi(G_{\Pi})}\sum_{u=1}^{\chi(G_{\Pi})}\prod_{v_i\in\mathcal{I}_u}\sum_{P_{i,j}\in \mathcal{T}_i}p_{ij}\frac{\id+P_{i,j}}{2}
\end{aligned}
\end{equation}
as we did in the construction of $\Omega_{\text{chr}}(\mathcal{S})$ for stabilizer codes. However, with a similar analysis, we can find that the maximum eigenvalue 
\begin{equation}
\lambda(\Omega_{\text{chr}}^{(1)})=\frac{1}{\chi(G_{\Pi})}\sum_{u=1}^{\chi(G_{\Pi})}\prod_{v_i\in\mathcal{I}_u}\frac{1}{2}\left(1+\frac{1}{a_i}\right)
\end{equation}
and the spectral gaps
\begin{equation}
\begin{gathered}
\Delta_{\min}(\Omega_{\text{chr}}^{(1)})=\frac{1}{\chi(G_{\Pi})}\min_{i_0}\left(\frac{1}{a_{i_0}}\prod_{i:\substack{\mathcal{I}_u\ni v_{i_0}\\v_i\in\mathcal{I}_u,i\neq i_0}}\frac{1}{2}\left(1+\frac{1}{a_i}\right)\right),\\
\Delta_{\max}(\Omega_{\text{chr}}^{(1)})=\frac{1}{\chi(G_{\Pi})}\sum_{u=1}^{\chi(G_{\Pi})}\left(\prod_{v_i\in\mathcal{I}_u}\frac{1}{2}\left(1+\frac{1}{a_i}\right)-\prod_{v_i\in\mathcal{I}_u}\frac{1}{2}\left(1-\frac{1}{a_i}\right)\right).
\end{gathered}
\end{equation}
Roughly speaking, as long as $a_i>1$ for most of $i$, $\lambda(\Omega_{\text{chr}}^{(1)})$, $\Delta_{\min}(\Omega_{\text{chr}}^{(1)})$ and $\Delta_{\max}(\Omega_{\text{chr}}^{(1)})$ go exponentially small with $m$. In general, since $\Delta_{\min}(\Omega_{\text{chr}}^{(1)})\leq \frac{1}{\chi(G_{\Pi})}\min_{u}\prod_{v_i\in\mathcal{I}_u}\frac{1}{2}\left(1+\frac{1}{a_i}\right)\leq\frac{1}{\chi(G_{\Pi})}\lambda(\Omega_{\text{chr}}^{(1)})$,
\begin{equation}
\begin{aligned}
N&\sim \frac{8\lambda(\Omega_{\text{chr}}^{(1)})(1-\lambda(\Omega_{\text{chr}}^{(1)}))\ln{1/\delta}}{\left(\Delta_{\min}(\Omega_{\text{chr}}^{(1)})-\Delta_{\max}(\Omega_{\text{chr}}^{(1)})\tau\right)^2 \epsilon^2}\\
&\geq \frac{8\lambda(\Omega_{\text{chr}}^{(1)})}{(\Delta_{\min}(\Omega_{\text{chr}}^{(1)}))^2}\cdot\frac{(1-\lambda(\Omega_{\text{chr}}^{(1)}))\ln{1/\delta}}{\epsilon^2}\\
&\geq \frac{8\chi(G_{\Pi})}{\Delta_{\min}(\Omega_{\text{chr}}^{(1)})}\cdot\frac{(1-\lambda(\Omega_{\text{chr}}^{(1)}))\ln{1/\delta}}{\epsilon^2}.
\end{aligned}
\end{equation}
Therefore, $N$ will scale exponentially with $m$ as long as $\Delta_{\min}(\Omega_{\text{chr}}^{(1)})$ goes exponentially small with $m$ and $\lambda(\Omega_{\text{chr}}^{(1)})$ is not negligible, which indicates that the construction of $\Omega_{\text{chr}}^{(1)}$ is usually useless.

\section{Proof of Proposition \ref{prop:imperfectomega}}\label{app:proofprop2}
First, one can obtain the passing probability 
\begin{equation}
\lambda-\Delta_{\max}(\Omega)\epsilon_{\rho} \leq \Tr{\Omega\rho}\leq \lambda-\Delta_{\min}(\Omega)\epsilon_{\rho}
\end{equation}
as we have done in Lemma \ref{lemma:passprob}. Consequently, the expected passing number shows 
\begin{equation}
\mathbb{E}[N_{\text{pass}}]\in[\left(\lambda-\Delta_{\max}(\Omega)\epsilon_{\rho}\right)N,\left(\lambda-\Delta_{\min}(\Omega)\epsilon_{\rho}\right)N].
\end{equation}

Let $p'_0$ be the solution to the equation $e^{-D\left[p'_0\|\lambda-\tau\Delta_{\max}(\Omega)\right]N}=e^{-D\left[p'_0\|\lambda-\Delta_{\min}(\Omega)\right]N}$. One can solve it to obtain 
\begin{equation}\label{eq:pprime0}
p'_0=\frac{\ln{\frac{1-\lambda+\Delta_{\max}(\Omega)\tau\epsilon}{1-\lambda+\Delta_{\min}(\Omega)\epsilon}}}{\ln{\frac{1-\lambda+\Delta_{\max}(\Omega)\tau\epsilon}{1-\lambda+\Delta_{\min}(\Omega)\epsilon}}+\ln{\frac{\lambda-\Delta_{\min}(\Omega)\epsilon}{\lambda-\Delta_{\max}(\Omega)\tau\epsilon}}}.
\end{equation}
If the device is ``good", there should be $\epsilon_{\rho}\leq \tau\epsilon$. By adopting the Chernoff–Hoeffding theorem, we can derive
\begin{equation}
\begin{aligned}
\mathrm{Pr}\left(N_{\text{pass}} \leq p'_0N\right)\leq e^{-D\left[p'_0\|\lambda-\tau\Delta_{\max}(\Omega)\epsilon\right]N}.
\end{aligned}
\end{equation}
based on
\begin{equation}\label{eq:geqprime}
\frac{\mathbb{E}[N_{\text{pass}}]}{N}\geq \lambda-\Delta_{\max}(\Omega)\epsilon_{\rho}\geq \lambda-\tau\Delta_{\max}(\Omega)\epsilon\geq p'_0.
\end{equation}
One can check the final inequality in Eq.~\eqref{eq:geqprime} given $p'_0$ in Eq.~\eqref{eq:pprime0} when $r>1$.

If the device is ``bad", there should be $\epsilon_{\rho}\geq \epsilon$. We can similarly derive
\begin{equation}
\begin{aligned}
\mathrm{Pr}\left(N_{\text{pass}} \geq p'_0N\right)\leq e^{-D\left[p'_0\|\lambda-\Delta_{\min}(\Omega)\epsilon\right]N}.
\end{aligned}
\end{equation}
based on
\begin{equation}
\frac{\mathbb{E}[N_{\text{pass}}]}{N}\leq \lambda-\Delta_{\min}(\Omega)\epsilon_{\rho}\leq \lambda-\Delta_{\min}(\Omega)\epsilon\leq p'_0.
\end{equation}
The final inequality can also be derived with $p'_0$ in Eq.~\eqref{eq:pprime0} when $r>1$.

Finally, the definition of $\delta$ gives $\delta=e^{-D\left[p'_0\|\lambda-\tau\Delta_{\max}(\Omega)\right]N}=e^{-D\left[p'_0\|\lambda-\Delta_{\min}(\Omega)\right]N}$. Therefore, 
\begin{equation}
N=\frac{\ln{1/\delta}}{D\left[p'_0\|\lambda-\Delta_{\min}(\Omega)\epsilon\right]}
\end{equation}
When $\epsilon$ is sufficiently small, we can simplify it to the second order in $\epsilon$: 
\begin{equation}
\begin{aligned}
N\sim \frac{8\lambda(1-\lambda)\ln{1/\delta}}{\left(\Delta_{\min}(\Omega)-\Delta_{\max}(\Omega)\tau\right)^2 \epsilon^2}.
\end{aligned}
\end{equation}

\section{Details of direct fidelity estimation and proof of Theorem \ref{theo:QSV+DFE}}\label{app:prooftheo2}

To construct an unbiased estimator $Y$ for $\frac{1}{2^k}\sum_{\bar{P}}\Tr{\bar{P}\ketbra{\psi}{\psi}}\cdot\Tr{\bar{P}\rho}$, we start by establishing that the quantity $p(\bar{P}):=\frac{1}{2^k} \left(\Tr{\bar{P} \ketbra{\psi}{\psi}}\right)^2$ satisfies the condition $\sum_{\bar{P}} p(\bar{P})=1$. Consequently, $p(\bar{P})$ can be treated as a probability distribution. Also, define $l:=\frac{8}{\epsilon_2^2 \delta_2}$. The first step is to sample $l$ logical Pauli operators $\left\{\bar{P}_i\right\}_{i=1}^l$ according to the probability distribution $p(\bar{P})$. For each sampled Pauli operator $\bar{P}_i$, define $m_i := \frac{\delta_2 \log(4/\delta_2)}{2^k p(\bar{P})}$. We then measure the corresponding operator $\bar{P}_i$ on the input states $m_i$ times, obtaining measurement outcomes, denoted as $A_{ij}$. The estimator $Y$ is a post-processing on the measurement outcomes, 
\begin{equation}
Y=\frac{1}{l}\sum_i\left(\frac{1}{\Tr{\bar{P} \ketbra{\psi}{\psi}}}\cdot\frac{1}{2^{k/2} m_i}\sum_j A_{ij}\right).
\end{equation}
Note that $\mathbb{E}_j A_{ij}=2^{k/2} \Tr{\bar{P}\rho}$. One can check that $Y$ is unbiased since
\begin{equation}
\begin{aligned}
\mathbb{E}_{i,j}[Y]&=\sum_{\bar{P}} p(\bar{P})\left(\frac{1}{\Tr{\bar{P} \ketbra{\psi}{\psi}}}\cdot\frac{1}{2^{k/2} m_i}\cdot m_i\cdot2^{k/2} \Tr{\bar{P}\rho}\right)\\
&=\frac{1}{2^k}\sum_{\bar{P}}\Tr{\bar{P}\ketbra{\psi}{\psi}}\cdot\Tr{\bar{P}\rho}.
\end{aligned}
\end{equation}

Since $(\id-\mathcal{P}_{\mathcal{V}})\rho(\id-\mathcal{P}_{\mathcal{V}})=\mathcal{P}_{\mathcal{V}_{\perp}}\rho\mathcal{P}_{\mathcal{V}_{\perp}}$ is non-negative, it could be composed to $(\id-\mathcal{P}_{\mathcal{V}})\rho(\id-\mathcal{P}_{\mathcal{V}})=\sum_{j}\lambda_j\ketbra{\phi_j}{\phi_j}$ with $\sum_{j}\lambda_j=1-\Tr{\mathcal{P}_{\mathcal{V}}\rho}$. Thus, if $\Tr{\mathcal{P}_{\mathcal{V}}\rho}\geq1-\epsilon_1$, 
\begin{equation}\label{eq:BCtoD}
\begin{aligned}
&\left|\Tr{\ketbra{\psi}{\psi}\rho}-\frac{1}{2^k}\sum_{\bar{P}}\Tr{\bar{P}\ketbra{\psi}{\psi}}\cdot\Tr{\bar{P}\rho}\right|\\
=&\left|\frac{1}{2^k}\Tr{\sum_{\bar{P}}\Tr{\bar{P}\ketbra{\psi}{\psi}}\bar{P}\cdot\left(\id-\mathcal{P}_{\mathcal{V}}\right)\rho\left(\id-\mathcal{P}_{\mathcal{V}}\right)}\right|\\
=&\left|\frac{1}{2^k}\sum_{\bar{P}}\Tr{\bar{P}\ketbra{\psi}{\psi}}\Tr{\bar{P}\sum_{j}\lambda_j\ketbra{\phi_j}{\phi_j}}\right|\\
\leq&\frac{1}{2^k}\sum_j\lambda_j\left|\sum_{\bar{P}}\Tr{\bar{P}\ketbra{\psi}{\psi}}\Tr{\bar{P}\ketbra{\phi_j}{\phi_j}}\right|\\
\leq&\frac{1}{2^k}\sum_j\lambda_j\left(\sum_{\bar{P}}\left(\Tr{\bar{P}\ketbra{\psi}{\psi}}\right)^2\cdot\sum_{\bar{P}}\left(\Tr{\bar{P}\ketbra{\phi_j}{\phi_j}}\right)^2\right)^{1/2}\\
\leq&\frac{1}{2^k}\epsilon_1\cdot\left(2^k\cdot 2^k\right)^{1/2}=\epsilon_1.
\end{aligned}
\end{equation}
The second line is based on the decomposition 
\begin{equation}
\ketbra{\psi}{\psi}=\mathcal{P}_{\mathcal{V}}\frac{1}{2^k}\cdot\sum_{\bar{P}}\Tr{\bar{P}\ketbra{\psi}{\psi}}\bar{P}
\end{equation}
and the fact that $\left(\id-\mathcal{P}_{\mathcal{V}}\right)^2=\id-\mathcal{P}_{\mathcal{V}}$. The fifth line is directly applying Cauchy–Schwarz inequality. The final line can be derived from the following fact.

\begin{fact}
Suppose $\bar{P}$ are logical Pauli operators of a $[[n,k,d]]$ QEC code, it can be transformed to Pauli operators of the first $k$ qubits with a decoding unitary $U$. To be specific, there exists a unitary $U$ such that $U\bar{P}U^{\dagger}=P_k\otimes \id_{n-k}$ where $P_k$ is a $k$-qubit Pauli operator. Therefore, for any $n$-qubit pure state $\ket{\phi}$, 
\begin{equation}
\begin{aligned}
\sum_{\bar{P}}\left(\Tr{\bar{P}\ketbra{\phi}{\phi}}\right)^2&=\sum_{\bar{P}}\left(\Tr{U\bar{P}U^{\dagger}U\ketbra{\phi}{\phi}U^{\dagger}}\right)^2\\
&=\sum_{P_k}\left(\Tr{P_k\otimes\id\cdot U\ketbra{\phi}{\phi}U^{\dagger}}\right)^2\\
&=\sum_{P_k}\left(\Tr{P_k\rho_k}\right)^2\\
&=2^k\Tr{\rho_{k}^2}\leq 2^k
\end{aligned}
\end{equation}
where $\rho_k=\text{Tr}_{n-k}\left(U\ketbra{\phi}{\phi}U^{\dagger}\right)$.
\end{fact}

For simplicity, we denote $N_{\text{pass}}\geq p_0N$ as event $A$, $\Tr{\mathcal{P}_{\mathcal{V}}\rho}\leq 1-\epsilon_1$ as event $B$, $\frac{1}{2^k}\sum_{\bar{P}}\Tr{\bar{P}\ketbra{\psi}{\psi}}\cdot\Tr{\bar{P}\rho}\leq Y-\epsilon_2$ as event $C$, and $\Tr{\ketbra{\psi}{\psi}\rho}\leq Y-\epsilon_1-\epsilon_2$ as event $D$. Eq.~\eqref{eq:BCtoD} shows that if both $B$ and $C$ are false, $D$ should also be false, which implies $\mathrm{Pr}(D)\leq \mathrm{Pr}(B\vee C)$. Therefore, we can conclude that $\mathrm{Pr}(A\land D)\leq\mathrm{Pr}(A\land(B\vee C))\leq \mathrm{Pr}(A\land B)+\mathrm{Pr}(A\land C)\leq \mathrm{Pr}(A|B)+\mathrm{Pr}(C)\leq \delta_1+\delta_2$, i.e., 
\begin{equation}\label{eq:finalfidel}
\mathrm{Pr}\left(N_{\text{pass}}\geq p_0N, \Tr{\ketbra{\psi}{\psi}\rho}\leq Y-\epsilon_1-\epsilon_2\right)\leq\delta_1+\delta_2.
\end{equation}
Suppose $\delta_1=\delta_2=\delta$ and $\epsilon_1=\epsilon_2=\epsilon$. Eq.~\eqref{eq:finalfidel} shows that if the quantum states $\rho$ pass the subspace verification, we can conclude that $\Tr{\ketbra{\psi}{\psi}\rho}\leq Y-2\epsilon$ with confidence $1-2\delta$. The sample complexity $N=N_1+N_2=O\left(\frac{\ln{1/\delta}}{\Delta_{\min}(\Omega)\epsilon}\right)+O\left(\frac{1}{\epsilon^2\delta}+\frac{2^k\ln{1/\delta}}{\epsilon^2}\right)$. 

\end{document}